\documentclass[final,3p,times, 11pt]{elsarticle}
\usepackage{lineno,hyperref}

\usepackage{amssymb}
\usepackage{subcaption}
\usepackage{float} 
\usepackage{ragged2e}
\usepackage{times}
\usepackage{amssymb}
\usepackage{graphicx}
\usepackage{enumerate}
\usepackage{amsmath}
\usepackage{mathtools}
\usepackage{algpseudocode}
\usepackage{algorithm}
\usepackage{multirow}
\usepackage[utf8]{inputenc}
\usepackage{footnote}
\usepackage{adjustbox}
\usepackage{amsthm}

\usepackage{lineno}

\newtheorem{theorem}{Theorem}[section]
\newtheorem{definition}{Definition}[section]
\newtheorem{lemma}[theorem]{Lemma}

\newtheorem{corollary}[theorem]{Corollary}

\begin{document}

\begin{frontmatter}



\title{Distributed Approximation Algorithms for Steiner Tree in the \\ $\mathcal{CONGESTED}$ \ $\mathcal{CLIQUE}$ model}


\author{Parikshit Saikia$^{1}$ and Sushanta Karmakar$^{1}$}
\address{\texttt{s.parikshit@iitg.ac.in, sushantak@iitg.ac.in}\\
		$^{1}$Department of Computer Science and Engineering\\
		Indian Institute of Technology Guwahati, India, 781039\\
}

\begin{abstract}
The \emph{Steiner tree} problem is one of the fundamental and classical problems in combinatorial optimization. In this paper we study this problem in the $\mathcal{CONGESTED}$ $\mathcal{CLIQUE}$ model of distributed computing and  present two deterministic distributed approximation algorithms for the same. The first algorithm computes a Steiner tree in $\tilde{O}(n^{1/3})$ rounds and $\tilde{O}(n^{7/3})$ messages for a given connected undirected weighted graph of $n$ nodes. Note here that $\tilde{O}(\cdot)$ notation hides polylogarithmic factors in $n$. The second one computes a Steiner tree in $O(S + \log\log n)$ rounds and $O(S (n - t)^2 + n^2)$ messages, where $S$ and $t$ are the \emph{shortest path diameter} and the number of \emph{terminal} nodes respectively in the given input graph. Both the algorithms admit an approximation factor of $2(1 - 1/\ell)$,  where $\ell$ is the number of terminal leaf nodes in the optimal Steiner tree.  For  graphs with $S = \omega(n^{1/3} \log n)$, the first algorithm exhibits  better performance than the second one in terms of the round complexity. On the other hand, for  graphs with  $S = \tilde{o}(n^{1/3})$,  the second algorithm outperforms the first one in terms of the round complexity. In fact when $S = O(\log\log n)$ then the second algorithm admits a round complexity of $O(\log\log n)$ and message complexity of $\tilde{O}(n^2)$. To the best of our knowledge, this is the first work to study the Steiner tree problem in the $\mathcal{CONGESTED}$ \ $\mathcal{CLIQUE}$ model.
\end{abstract}



\begin{keyword}
Steiner tree, shortest path forest, $\mathcal{CONGESTED}$ \ $\mathcal{CLIQIUE}$, distributed approximation algorithm



\end{keyword}

\end{frontmatter}


\section{Introduction}
The $\mathcal{CONGESTED}$ \ $\mathcal{CLIQUE}$ model (CCM) is one of the fundamental models in distributed computing that was first introduced by Lotker et al. \cite{Lotker:2005:MST:1085579.1085591}. In this model nodes can communicate with each other via an underlying communication network, which is a \emph{clique}. Communication happens in synchronous rounds and a pair of nodes can exchange $b$ bits in a round.  In this paper we assume that $b = O(\log n)$, where $n$ is the number of nodes in the communication network. In literature there are two other classic models of distributed computing, namely $\mathcal{LOCAL}$ and $\mathcal{CONGEST}$. The $\mathcal{LOCAL}$ model of distributed computing mainly focuses on \emph{locality}\footnote{In distributed computing the \emph{locality} means processors are restricted to collecting data from others which are at a distance of $x$ hops in $x$ time units. The issue of locality is, to what extent a global solution to a computational problem can be obtained from locally available data \cite{Linial:1992:LDG:130563.130578}.} and ignores the congestion by allowing messages of unlimited sizes to be communicated \cite{Peleg2000ANL,Linial:1992:LDG:130563.130578}. The $\mathcal{CONGEST}$ model on the other hand simultaneously considers the congestion (by bounding the transmitted message size) and locality. In contrast, the CCM takes locality out of the picture and solely focuses on congestion. Since in CCM the hop diameter is \emph{one} therefore nodes can directly communicate with each other and in each round they can together exchange $O(n^2 \log n)$ bits. Note that in all the above three models of distributed computing, nodes (processors) are considered computationally unbounded.

The $n$ nodes and a subset of the edges of the communication network form the \emph{input graph} in the CCM. In the input graph each node has a unique identity (ID) and it knows the weights of all the edges incident on it. In this paper we study a classical combinatorial optimization problem, the \emph{Steiner tree} problem, in the CCM. It is defined as follows.

\begin{definition}[{Steiner tree (ST) problem}]
Given a connected undirected graph $G=(V,E)$ and a weight function $w : E \rightarrow \mathbb{R}^{+}$, and a set of vertices $Z \subseteq V$, known as the set of terminals, the goal of the ST problem is to find a tree $T'=(V',E')$ such that $\sum_{e \in E'} w(e)$ is minimized subject to  $Z \subseteq V' \subseteq V$ and $E' \subseteq E$.
\end{definition}

The set $V \setminus Z$ is known as the set of \emph{non-terminals} or \emph{Steiner} nodes. Note that if $|Z| = 2$ then the ST problem reduces to the problem of finding {\em shortest-path} between two distinct nodes in the network. On the other hand if $|Z| = |V|$ then the ST problem becomes the minimum spanning tree (MST) problem. Specifically the ST problem is a generalized version of the MST problem. It is known that both MST and the shortest path problem can be solved in polynomial time
whereas the ST problem is one of the original 21 problems proved NP-complete by Karp \cite{DBLP:conf/coco/Karp72} (in the centralized setting). The best known (polynomial time) approximation ratio  for solving the ST problem in the centralized setting is $\ln (4) + \epsilon \approx 1.386 + \epsilon$, for $\epsilon > 0$ due to Byrka et al. \cite{Byrka:2010:ILA:1806689.1806769}. It is also known that  the ST problem  can not be solved in polynomial time with an approximation factor  $\leq \frac{96}{95}$  unless $P=NP$
\cite{chlebik:2008:STP:1414105.1414423}. 

The ST problem finds applications in numerous areas such as the VLSI layout design, communication networks, transportation networks, content distribution (video on demand, streaming multicast) networks, phylogenetic tree reconstruction in computational biology etc. Moreover the ST problem appears as a subproblem or as a special case of many other  problems in network design such as Steiner forest, prize-collecting Steiner tree etc. There are many variations of the ST problem such as directed Steiner tree, metric Steiner tree, euclidean Steiner tree, rectilinear Steiner tree, and so on. Hauptmann and Karpinaski \cite{Hauptman-Karpinaski} provide a website with continuously updated state of the art results for many variants of the problem.

\vspace{2.5em}
\noindent
{\bf Motivation.} The  motivation behind the study of the CCM is to understand the role of congestion in distributed computing. There has been a lot of progress in solving various problems in the CCM including minimum spanning tree (MST) \cite{Lotker:2005:MST:1085579.1085591,10.1007/978-3-662-45174-8_35,Ghaffari:2016:MLR:2933057.2933103,Jurdzinski:2018:MOR:3174304.3175472}, facility location \cite{Berns:2012:SDA:2359888.2359973,Gehweiler2014}, shortest paths and distances \cite{Censor-Hillel:2015:AMC:2767386.2767414,holzer_et_al:LIPIcs:2016:6597,Nanongkai:2014:DAA:2591796}, subgraph detection \cite{Drucker:2012:CCD:2332432.2332443}, triangle finding \cite{Drucker:2012:CCD:2332432.2332443, Dolev:2012:TTA:2427873.2427892}, sorting \cite{Patt-Shamir:2011:RCD:1993806.1993851,Lenzen:2013:ODR:2484239.2501983}, routing \cite{Lenzen:2013:ODR:2484239.2501983}, and ruling sets \cite{10.1007/978-3-662-45174-8_35,Berns:2012:SDA:2359888.2359973}. Recently Pai and Pemmaraju \cite{DBLP:journals/corr/abs-1905-09016} studied the \emph{graph connectivity} lower bounds in the CCM. Specifically they showed that the lower bound round complexity for graph connectivity problem in the BCCM($1$)\footnote{In literature the CCM is classified into two types namely $\mathcal{BROADCAST}$ $\mathcal{CONGESTED}$ \ $\mathcal{CLIQUE}$ model (BCCM) and $\mathcal{UNICAST}$ \ $\mathcal{CONGESTED}$ \ $\mathcal{CLIQUE}$ model (UCCM) \cite{Drucker:2014:PCC:2611462.2611493}. In BCCM($b$) each node can only broadcast a single $b$-bit message over each of its incident links in each round. On the other hand the $b$-bits UCCM (UCCM($b$)) allows each node to send a possibly different $b$-bit message over each of its incident links in the network in each round.}, which is $\Omega(\log n)$, holds for both deterministic as well as constant-error randomized Monte Carlo algorithms \cite{DBLP:journals/corr/abs-1905-09016}. Despite the fact that the ST problem has been extensively studied in the $\mathcal{CONGEST}$ model of distributed computing \cite{GENHUEY199373,PC_JF_2005,MK_FK_DM_GP_KT_2008,Lenzen:2015:FPD,Saikia-Karmakar2019}, to the best of our knowledge, such a study has not been carried out in the CCM. The best deterministic round complexity for solving the ST problem in the $\mathcal{CONGEST}$ model was recently proposed by Saikia and Karmakar \cite{Saikia-Karmakar2019}, which is $O(S + \sqrt{n} \log ^* n)$ with the approximation factor of $2(1 - 1/\ell)$, where $S$ is the \emph{shortest path diameter}\footnote{The term \emph{shortest path diameter} was first introduced by Khan and Pandurangan \cite{Khan2008}.} (the definition is deferred to Section~\ref{model-notations}) of a graph and $\ell$ is the number of terminal leaf nodes in the optimal ST, which improves the previous best round complexity of the ST problem \cite{Lenzen:2015:FPD}.  The MST problem is a special case of the ST problem and has been extensively studied in the $\mathcal{CONGEST}$ model as well as in the CCM. The best deterministic round complexity known so far for solving the MST problem in the CCM  is due to Lotker et al. \cite{Lotker:2005:MST:1085579.1085591}, which is $O(\log \log n)$. The algorithm in \cite{Lotker:2005:MST:1085579.1085591} has a message complexity of $O(n^2)$. There are other algorithms for the MST problem in the CCM that are randomized in nature with the round complexities of $O(\log \log \log n)$ \cite{Hegeman:2015:TOB:2767386.2767434} and $O(\log ^* n)$ \cite{pemmaraju_et_al:LIPIcs:2016:6882,Ghaffari:2016:MLR:2933057.2933103}. The message complexities of the algorithms in \cite{Hegeman:2015:TOB:2767386.2767434}, \cite{Ghaffari:2016:MLR:2933057.2933103}, and \cite{pemmaraju_et_al:LIPIcs:2016:6882} are $O(n^2)$, $O(n^2)$, and $o
(m)$ respectively. Here $m = |E|$. Recently Jurdzi\'{n}ski and Nowicki \cite{Jurdzinski:2018:MOR:3174304.3175472} achieved a randomized algorithm that constructs an MST in $O(1)$ rounds in the CCM with high probability. Therefore an intriguing question is:

\medskip
\noindent
{\em ``What is the best round complexity that can be achieved in solving the ST problem in the CCM while maintaining an approximation factor of at most $2$?''}

\medskip
In CCM, one can trivially compute an ST in $O(n)$ rounds by maintaining an approximation factor of at most $2$. It can be computed as follows. One can collect the entire topology of the input graph in a special node $r$, which takes $O(n)$ rounds. Then we can compute an ST by applying one of the best known centralized ST algorithms \cite{Kou:1981:FAS:2697742.2698014,Wu1986,Byrka:2010:ILA:1806689.1806769} whose approximation factor is at most $2$, and finally inform each of the nodes involved with the resultant ST. Note that the resultant ST has at most $n - 1$ edge information which can be decomposed into $O(n)$ messages. Therefore $r$ can perform the final step in $O(1)$ rounds by sending each edge of the resultant ST to a different intermediate node, which will eventually sends to the destined node.

\bigskip
\noindent
{\bf Our contribution.} In this work we propose two non-trivial deterministic distributed  approximation algorithms for the ST problem in the CCM. Both the algorithms admit an approximation factor of $2(1 - 1/\ell)$. The first one, which will be denoted by STCCM-A, computes an  ST in $\tilde{O}(n^{1/3})$ rounds and $\tilde{O}(n^{7/3})$ messages. We also propose a deterministic distributed \emph{shortest path forest} (SPF) (the definition is deferred to Section~\ref{description-spf}) algorithm in the CCM (henceforth it will be denoted by SPF-A) that computes a SPF in $O(n^{1/3} \log n)$ rounds and $O(n^{7/3} \log n)$ messages which will be used as a subroutine in the proposed STCCM-A algorithm. The SPF-A algorithm is based on an \emph{all pairs shortest path} (APSP) algorithm in the CCM due to Censor-Hillel et al. \cite{Censor-Hillel:2015:AMC:2767386.2767414}. The first contribution of this paper is summarized in the following theorem.

\begin{theorem}
Given a connected undirected weighted graph $G = (V, E, w)$ and a terminal set $Z \subseteq V$, there exists an algorithm that computes an ST in $\tilde{O}(n^{1/3})$ rounds and $\tilde{O}(n^{7/3})$ messages  in the CCM with an approximation factor of $2(1 - 1/\ell)$, where $\ell$ is the number of terminal leaf nodes in the optimal ST. 
\end{theorem}

The proposed STCCM-A algorithm is inspired by an algorithm proposed in \cite{Saikia-Karmakar2019}. It consists of four steps (each step is a small distributed algorithm)\--- the first step is to build a SPF $G_F$ of the input graph $G$ for a given terminal set $Z$, which is essentially a partition of the graph $G$ into disjoint trees: Each partition contains exactly one terminal and a subset of non-terminals. A non-terminal $v$ joins a partition containing the terminal $z\in Z$ if and only if $\forall x \in Z \setminus \{z\}$, $d(z, v) \leq d(x, v)$.\footnote{$d(u, v)$ denotes the (weighted) shortest distance between nodes $u$ and $v$ in graph $G$.}  In the second step the weights of the edges of $G$ with respect to the SPF $G_F$ are suitably changed, which produces a modified graph $G_c$; in the third step the MST algorithm proposed by Lotker et al. \cite{Lotker:2005:MST:1085579.1085591} is applied on the graph $G_c$ to build an MST $T_M$; and finally some edges are pruned from $T_M$ in such a way that in the remaining tree $T_Z$ (which is the required ST) all leaves are terminal nodes.

The second proposed algorithm for the ST problem in the CCM, which will be denoted by STCCM-B, computes a $2(1 - 1/\ell)$-approximate ST in $O(S + \log \log n)$ rounds and $O(S(n - t)^2   + n^2)$ messages. Similar to the STCCM-A algorithm, the  STCCM-B algorithm also consists of four steps. Except the step $1$, all other steps in STCCM-B algorithm are same as that of the STCCM-A algorithm. For step $1$ in the STCCM-B algorithm, which is the SPF construction, we propose another SPF algorithm in the CCM (henceforth it will be denoted by SPF-B) that computes a SPF in $O(S)$ rounds and $O(S (n - t)^2 +  nt)$ messages. The second contribution of this paper is summarized in the following theorem.

\begin{theorem}
Given a connected undirected weighted graph $G = (V, E, w)$ and a terminal set $Z \subseteq V$, there exists an algorithm that computes an ST in $O(S + \log \log n)$ rounds and $O(S (n - t)^2 + n^2)$ messages  in the CCM with an approximation factor of $2(1 - 1/\ell)$, where $n > t$, $S$ is the shortest path diameter of \ $G$, and \ $\ell$ is the number of terminal leaf nodes in the optimal ST. 
\end{theorem}

As a by-product of the above theorem, for constant or sufficiently small shortest path diameter networks (where $S = O(\log \log n)$) the following corollary holds.
\begin{corollary} \label{corolloary-1}
If \ $S = O(\log \log n)$ then a $2(1 - \frac{1}{\ell})$-approximate ST can be deterministically computed in $O(\log \log n)$ rounds and $\tilde{O}(n^2)$ messages in the CCM.
\end{corollary}

\noindent
The above round and message complexities of the STCCM-B algorithm almost coincide with the best known deterministic result for MST construction in the CCM due to Lotker et al. \cite{Lotker:2005:MST:1085579.1085591} and the approximation factor of the resultant ST is at most $2(1 - 1/\ell)$ of the optimal.

\vspace{2em}
\noindent
{\bf Related work.} Chen et al. \cite{GENHUEY199373} proposed the first deterministic distributed algorithm for the ST problem in the  $\mathcal{CONGEST}$ model and achieved an approximation factor of $2(1 - 1/\ell)$. The time and message complexities of this algorithm  are $O(n(n - t))$ and $O(m + n(n - t + \log n))$ respectively. Chalermsook et al. \cite{PC_JF_2005} presented a deterministic distributed $2$-approximate algorithm for the ST problem in the synchronous model (which allows a bounded message size only) with time and message complexities of $O(n \log n)$ and $O(tn^2)$ respectively. Khan et al. \cite{MK_FK_DM_GP_KT_2008} presented an $O(\log n)$-approximate randomized distributed algorithm for the ST problem in the $\mathcal{CONGEST}$ model with time complexity $O(S \log^2 n)$ and message complexity $O(Sn \log n)$. Lenzen and Patt-Shamir \cite{Lenzen:2015:FPD} presented two distributed algorithms for the Steiner forest problem (a more generalized version of the ST problem) in the  $\mathcal{CONGEST}$ model: one is deterministic and the other one is randomized. The former one finds, a $(2 + o(1))$-approximate Steiner forest in $\tilde{O}(\sqrt{\min\{D, k\}} (D + k) + S + \sqrt{\min\{St, n\}})$ rounds, where $D$ is the unweighted diameter and $k$ is the number of terminal components in the given input graph. The latter one finds a $(2 + o(1))$-approximate Steiner forest in $\tilde{O}(\sqrt{\min\{D, k\}} (D + k)  + \sqrt{n})$ rounds with high probability. Note that if $k=1$ then the Steiner forest problem reduces to the ST problem. In this case the round complexities of the two algorithms in \cite{Lenzen:2015:FPD}, in which one is deterministic and the other one is randomized reduce to $\tilde{O}(S + \sqrt{\min\{St, n\}})$ and $\tilde{O}(D  + \sqrt{n})$ respectively. Saikia and Karmakar \cite{Saikia-Karmakar2019} proposed a deterministic distributed $2(1 - 1/\ell)$-factor approximation algorithm for the ST problem in the $\mathcal{CONGEST}$ model with the round and message complexities of $O(S + \sqrt{n} \log ^* n)$ and $O(\Delta(n - t)S + n^{3/2})$ respectively, where $\Delta$ is the maximum degree of a node in the graph. Recently Bachrach et al.  \cite{DBLP:journals/corr/abs-1905-10284} showed that the lower bound round complexity in solving the ST problem exactly in the $\mathcal{CONGEST}$ model is $\Omega(n^2/\log^2 n)$. In the approximate sense on the other hand, no immediate lower bound result exists for solving the ST problem in the $\mathcal{CONGEST}$ model. Since the ST problem is a generalized version of the MST problem, we believe that in the approximate sense the lower bound results for the MST problem\footnote{Das Sarma et al. \cite{DasSarma:2011:DVH:1993636.1993686} showed that approximating (for any factor $\alpha \geq 1$) MST requires $\Omega(D + \sqrt{n/(B \log n)})$ rounds (assuming $B$ bits can be sent through each edge in each round) in the $\mathcal{CONGEST}$ model. Kutten et al. \cite{Kutten:2015:CUL:2742144.2699440} established that $\Omega(m)$ is the message lower bound  for leader election in the $KT_0$ model (i.e. {\bf K}nowledge {\bf T}ill radius {\bf 0}) which holds for both the deterministic as well as randomized (Monte Carlo) algorithms even if the network parameters $D$, $n$, and $m$ are known, and all the nodes wake up simultaneously. Since a distributed MST algorithm can be used to elect a leader, the above message lower bound in the $KT_0$ model also applies to the distributed MST construction.} also hold for the ST problem in the $\mathcal{CONGEST}$ model.  
\begin{table}[t]
\footnotesize
        \centering
        \begin{tabular}{|c|c|c|c|c|c|c|}
            \hline
             {\bf References} & {\bf Model} & {\bf Type} & {\bf Round complexity} & {\bf Message complexity} & {\bf Approximation}\\
             \hline
             & & & & & \\
            Chen et al. \cite{GENHUEY199373} & CM & DT & $O(n(n -t))$ & $O(m + n(n - t + \log n))$ & $2(1 - 1/\ell)$\\ 
            & & & & & \\
            \hline
             & & & & & \\
            Chalermsook et al. \cite{PC_JF_2005} & CM & DT & $O(n \log n)$ & $O(tn^2)$ & $2$\\ 
           & & & & & \\
            \hline
             & & & & & \\
            Khan et al. \cite{MK_FK_DM_GP_KT_2008} & CM & RM & $O(S \log^2 n)$ & $O(Sn \log  n)$ & $O(\log n)$ \\
           & & & & & \\
           
           \hline
                       
             & & & & & \\
            \multirow{4}{*}{Lenzen and Patt-Shamir \cite{Lenzen:2015:FPD}} & \multirow{4}{*}{CM} & DT & $\tilde{O}(S + \sqrt{min\{st, n\}})$ & - & $2 + o(1)$ \\ 
           & & & & & \\
            \cline{3-6}
            & & & & & \\
            & & RM & $\tilde{O}(D + \sqrt{n})$ & - & $2 + o(1)$ \\
             & & & & & \\
            \hline
             & & & & & \\
            Saikia and Karmakar \cite{Saikia-Karmakar2019} & CM & DT & $O(S + \sqrt{n} \log ^* n)$ & $O(\Delta (n - t) S + n^{3/2})$ & $2(1 - 1/\ell)$ \\
             & & & & & \\
           \hline 
            & & & & & \\
            \multirow{4}{*}{{\bf this paper}} & \multirow{4}{*}{CCM} & \multirow{4}{*}{DT} & $\tilde{O}(n^{1/3})$ & $\tilde{O}(n^{7/3})$ & $2(1 - 1/\ell)$ \\
           & & & & & \\
            \cline{4-6}
            & & & & & \\
            & & & $O(S + \log\log n)$ & $O(S(n - t) ^2 + n^2)$ & $2(1 - 1/\ell)$ \\
             & & & & & \\
            \hline
         \end{tabular}
        \vspace{2em}
        \caption{\label{summary_related_works} Summary of results for Steiner tree problem in distributed setting. Here CM = $\mathcal{CONGEST}$ model, DT = deterministic, RM = randomized, $n = |V|$, $m = |E|$, $t = |Z|$, $n > t$, $S$ and $D$ are the shortest path diameter and the unweighted diameter respectively of a connected undirected weighted graph $G$.}
\end{table}    

\bigskip
Performances of some of the distributed algorithms mentioned above, together with that of our work, are summarized in Table~\ref{summary_related_works}.

\bigskip
\noindent
{\bf Paper organization.} The rest of the paper is organized as follows. In Section~\ref{model-notations} we define the system model and notations. The description of the SPF-A algorithm is given in Section~\ref{description-spf}. The description of the STCCM-A algorithm and an illustrative example of it are given in Section~\ref{stccm-algorithm}. The proof of the STCCM-A algorithm is given in Section~\ref{proof-stccm-algorithm}. The description and proof of the STCCM-B algorithm are given in Section~\ref{stccm-b-algorithm}. A brief description of the Censor-Hillel et al.'s APSP algorithm and  Lotker et al's MST algorithm are deferred to ~\ref{Censor-Hillel-all-algo} and ~\ref{Lotker-algo} respectively. We conclude the paper in Section~\ref{conclusion}.

\vspace{1em}
\section{Model and notations} \label{model-notations}
{\bf System model}. We consider the CCM as specified in \cite{Lotker:2005:MST:1085579.1085591}. This model consists of a complete network described by a \emph{clique} of $n$ nodes. Each node represents an independent computing entity (processor). Nodes are connected through a point-to-point network of $n \choose 2$ bidirectional communication links. The bandwidth of each of the communication links is bounded by $O(\log n)$ bits. 

At the beginning of the computation, each node knows its own (unique) identity (ID) which can be represented by $O(\log n)$ bits and the part of the input assigned to it.\footnote{In this paper, we assume that initially each node knows only its own identity and weights of all the edges incident on it; nodes do not have
initial knowledge of the identifiers of their neighbors and other nodes. This is known as the $KT_0$ model, also called the \emph{clean network model} \cite{Peleg_2000}. On the other hand if each node has initial knowledge of itself and the identifiers of its neighbors then such a model is known as the $KT_1$ model (i.e. {\bf K}nowledge {\bf T}ill radius {\bf 1}). In $KT_1$ model only the knowledge of the identifiers of neighbors is assumed, not other information such as the incident edges of the neighbors.} In the CCM the \emph{input graph} $G$ and the underlying \emph{communication network}, which is a clique, are not  same. The input graph $G$ is distributed among the nodes (processors) of the clique via a vertex partition. In this paper we consider that each vertex of $G$ along with its incident edges (with their weights) are assigned to a distinct node (processor) in the clique. If an edge does not exist in $G$ then the weight of such an edge is considered equal to $\infty$ in the clique. Communication happens in synchronous rounds and a pair of nodes can exchange $O(\log n)$ bits in each round.  Nodes communicate and coordinate their actions with other nodes by passing messages (of size $O(\log n)$ bits) only.  In general, a message contains a constant number of edge weights, node IDs, and other arguments (each of them is polynomially bounded in $n$). Note here that each of the arguments in a message is polynomially bounded in $n$ and therefore polynomially many sums of arguments can be encoded with $O(\log n)$ bits. We assume that nodes and links do not fail.

An execution of the system advances in synchronous rounds. In each round: nodes receive messages that were sent to them in the previous round, perform some local computation, and then send (possibly different) messages. The time complexity is measured by the number of rounds required until all the nodes (explicitly) terminate. The message complexity is measured by  the  number of messages sent until all the nodes (explicitly) terminate. 

\medskip
{\bf Notations.} We use the following terms and notations.

\begin{itemize}
\item $\mathit{\delta(v)}$ denotes the set of edges incident on a node $v$. Similarly $\delta(C)$ denotes the set of edges having exactly one endpoint in a subgraph $C$.
\item $w(e)$ denotes the weight of an edge $e$.
\item $\mathit{ES(e)}$ denotes the state of an edge $e$. 
\item $s(v)$ denotes the source node of a node $v$.
\item $d(v)$ denotes the {\em shortest distance} from a node $v$ to its source ID $s(v)$.
\item $\pi(v)$ denotes the predecessor node of a node $v$.
\item Let $M_{n \times n}$ denotes a matrix of size $n \times n$, where $n$ is a positive integer. Then $m(i, j)$ denotes the value of an entry located at the $i^{th}$ row and $j^{th}$ column in $M_{n \times n}$.
\item $\langle X \rangle$ denotes the message $X(a_1, a_2,...)$. Here $a_1, a_2,...$ are the arguments of the message $X$. Note that unless it is necessary, arguments of $\langle X \rangle$ will not be shown in it.
\end{itemize}

\section{SPF construction in CCM} \label{description-spf}
The SPF is defined as follows.
\begin{definition}[SPF \cite{GENHUEY199373}] \label{defi-spf}
Let $G = (V, E, w)$ be a connected undirected weighted graph, where $V$ is the vertex set, $E$ is the edge set, and $w : E \rightarrow \mathbb{R}^+$ is the non-negative weight function. Given a subset $Z \subseteq V$, a SPF for $Z$ in $G$ is a sub-graph $G_F = (V, E_F, w)$ of $G$ consisting of disjoint trees $T_i = (V_i, E_i, w)$, $i = 1, 2, ..., |Z|$ such that 
\begin{itemize}
	\item for all $i$, $V_i$ contains exactly one node $z_i$ of $Z$.
	\item if $v \in V_i$ then its  $s(v)$ in $G$ is $z_i \in Z$.
	\item $V_1 \cup V_2 \cup ... V_{|Z|} = V$ and $V_i \cap V_j = \phi$ for all $i \neq j$.
	\item $E_1 \cup E_2 \cup ... E_{|Z|} = E_F \subseteq E$.
	\item The shortest path between $v$ and $s(v) = z_i$ in $T_i$ is the shortest path between $v$ and $s(v)$ in $G$.
\end{itemize}
\end{definition}

\subsection{SPF-A algorithm} \label{spf-algo-outline}
\vspace{.5em}
The SPF-A algorithm is used as a subroutine in the proposed STCCM-A algorithm. Here we give a brief description of it. It consists of two parts: the first part constructs an APSP of the input graph $G$ and the second part constructs the required SPF from the graph formed by the APSP. Specifically, in the first part we apply the algebraic method due to Censor-Hillel et al. \cite{Censor-Hillel:2015:AMC:2767386.2767414}. Censor-Hillel et al. showed that in the CCM, the iterated squaring of the weight matrix over the \emph{min-plus semiring} \cite{Fischer:1971:BMM:1446293.1446319,Munro:1971:EDT:2598952.2599354} computes an APSP in $O(n^{1/3} \log n)$ rounds and $O(n^{7/3} \log n)$ messages. One of the fundamental applications of the APSP is the construction of the routing tables in a network. Specifically a routing table entry denoted by $R[u, v] = w \in V$ is a node such that $(u, w) \in E$ and $w$ lies on a shortest path from $u$ to $v$. Censor-Hillel et al. also showed that in the CCM the iterated squaring algorithm can be used to construct routing tables of a network as well. For the sake of completeness a brief description of the iterated squaring algorithm and the procedure of the routing table construction due to Censor-Hillel et al. is provided in \ref{Censor-Hillel-all-algo}.

Now we describe the second part of the SPF-A algorithm and show that it requires $O(1)$ rounds and $O(n)$ messages. From the first part we assume that each node in $V$ knows the shortest path distances to all other nodes in $V$ and its routing table entries $R$. Therefore, by using the shortest path distance information each node $v \in V \setminus Z$ can locally choose the closest terminal as its source node $s(v)$. Note that there may be more than one terminal with equal shortest distances for a given non-terminal. In this case, the non-terminal chooses the one with the smallest ID among all such terminals. Once a non-terminal $v$ chooses the closest terminal as its source node $s(v)$,  by using its own routing table $R$, it can also choose its parent node $\pi(v)$ with respect to $s(v)$. Whenever a non-terminal $v$ sets its $\pi(v)$, it also informs $\pi(v)$ that it has chosen $\pi(v)$ as its parent.  It is obvious that to establish the parent-child relationship between a pair of nodes $(\pi(v), v)$, where $s(v) = s(\pi(v))$, it requires $O(1)$ rounds and $O(1)$ messages. In this way each node $v \in V \setminus Z$ is connected to exactly one tree rooted at some source node $s(v)$. Here each source node is a terminal. Therefore exactly $|Z|$ number of shortest path trees are constructed by the above procedure which together form the required SPF. Since the procedure of choosing parent can be started in parallel by all nodes in $V \setminus Z$, for all such pair of nodes it requires $O(1)$ rounds and $O(n)$ messages. It is clear that the overall round and message complexities of the  SPF-A algorithm are dominated by the first part of the algorithm. Therefore the following theorem holds.

\begin{theorem}
\label{spf-time-complexity}
A SPF can be deterministically computed in $\tilde{O}(n^{1/3})$ rounds and $\tilde{O}(n^{7/3})$ messages in the CCM.
\end{theorem}

The correctness of the SPF-A algorithm directly follows from the correctness of the algorithm proposed by Censor-Hillel et al. \cite{Censor-Hillel:2015:AMC:2767386.2767414}.

\section{STCCM-A algorithm} \label{stccm-algorithm}
\subsection{Preliminaries} 

\begin{definition}(Complete distance graph)
A graph $K_Z$ is called a complete distance graph on the node set $Z \subseteq V$ of a connected undirected weighted graph $G = (V, E, w)$ only if for each pair of nodes $u, v \in Z$, there is an edge $(u, v)$ in $K_Z$ and the weight of the edge $(u, v)$ is the length of the shortest path between $u$ and $v$ in $G$.
\end{definition}

The approximation factor of the proposed STCCM-A algorithm directly follows from the correctness of a centralized algorithm  due to Kou et al. \cite{Kou:1981:FAS:2697742.2698014} (Algorithm H). For a given connected undirected weighted graph $G = (V, E, w)$ and a set of terminals $Z \subseteq V$, the Algorithm H computes an ST $T_Z$ as follows.
\begin{enumerate}
	\item Construct a complete distance graph $K_Z$.
	\item Find an MST $T_A$ of $K_Z$.
	\item Construct a sub-graph $G_A$ of $T_A$ by replacing each edge of $T_A$ with its corresponding shortest distance path in $G$.
	\item Find an MST $T_B$ of $G_A$.
	\item Construct an ST $T_Z$, from $T_B$ by deleting edges of $T_B$ so that all leaves of $T_Z$ are terminals.
\end{enumerate}
The running time of the Algorithm H is $O(tn^2)$. Following the principles of both Prim's and Krushkal's algorithms, Wu at el. \cite{Wu1986} proposed a faster centralized algorithm (Algorithm M) which improves the time complexity to $O(m \log n)$, achieving the same approximation ratio as that of the Algorithm H. The Algorithm M computes a sub-graph called {\em generalized MST} $T_Z$ for $Z$ of $G$ which is defined as follows.
\begin{definition}[Generalized MST \cite{Wu1986}]
Let $G = (V, E, w)$ be a connected undirected weighted graph and $Z$ be a subset of $V$. Then a generalized MST $T_Z$ is a sub-graph of $G$ such that
\begin{itemize}
	\item there exists an MST $T_A$ of the  complete distance graph $K_Z$ such that for each edge $(u, v)$ in $T_A$, the length of the unique path in between $u$ and $v$ in $T_Z$ is equal to the weight of the edge $(u, v)$ in $T_A$.
	\item all leaves of $T_Z$ are in $Z$.  
\end{itemize} 
\end{definition}
It is clear that $T_Z$ is an ST for $Z$ in $G$ and is the actual realization of $T_A$ (Recall that $T_A$ is the MST of $K_Z$). In summary, the Algorithm M constructs a generalized MST $T_Z$ for $Z$ of $G$ as follows. Initially, the set of nodes in $Z$ are treated as a forest of $|Z|$ separate trees and successively merge them until all of them are in a single tree $T_Z$. A priority queue $Q$ is used to store the frontier vertices of paths extended from the trees. Each tree gradually extends its branches into the node set $V \setminus Z$. When two branches  belonging to two different trees meet at some node then they form a path through that node merging these two trees. The algorithm always guarantees to compute only such paths of minimum length for merging trees.

The proposed STCCM-A algorithm also computes a generalized MST for $Z$ of $G$ (which is essentially the required ST with the approximation factor of $2(1 - 1/\ell)$)  in the CCM. The round and message complexities of the STCCM-A algorithm are $\tilde{O}(n^{1/3})$ and $\tilde{O}(n^{7/3})$ respectively. Saikia and Karmakar \cite{Saikia-Karmakar2019} proposed a $2(1 - 1/\ell)$-factor deterministic distributed  approximation algorithm which also computes a generalized MST for $Z$ of $G$ in $O(S + \sqrt{n} \log ^* n)$ rounds and $O(\Delta(n - t)S + n^{3/2})$ messages. However, the algorithm in \cite{Saikia-Karmakar2019} was proposed for the $\mathcal{CONGEST}$ model, whereas the STCCM-A algorithm is proposed for the CCM. 


There are four small distributed algorithms (step $1$ to step $4$) involved in the STCCM-A algorithm similar to the algorithm proposed in \cite{Saikia-Karmakar2019}. However,  except the step $2$, all other steps in the STCCM-A algorithm are different from the algorithm proposed in \cite{Saikia-Karmakar2019}. The step $2$ of the STCCM-A algorithm is similar to that of the algorithm proposed in \cite{Saikia-Karmakar2019}.

\subsection{Outline of the STCCM-A algorithm}

\noindent
{\bf Input specification.}
We assume that there is a special node $r \in V$ available at the start of the algorithm. For correctness we assume that $r$ is the node with the smallest ID in the system. Initially each node $v \in V$ knows its unique ID, whether it is a terminal or a non-terminal, and weight $w(e)$ of each edge $e \in \mathit{\delta(v)}$. Each node in $\mathit{V}$ maintains a boolean variable $\mathit{steiner\_flag}$ whose values can be either $\mathit{true}$ or $\mathit{false}$. Initially $\mathit{steiner\_flag(v)}$ is set to $\mathit{false}$ for each $\mathit{v \in V \setminus Z}$, whereas throughout the execution of the algorithm the value of $\mathit{steiner\_flag(v)}$ is $\mathit{true}$ for each $\mathit{v \in Z}$.  Also each node $v \in V$ initially sets its local variable $\mathit{ES(e)}$ to $\mathit{basic}$ for each $\mathit{e \in \delta(v)}$. Recall that $\mathit{ES(e)}$ denotes the state of an edge $e$.


\medskip
\noindent
{\bf Output specification.} Whenever the algorithm terminates, the pair $\mathit{(steiner\_flag, B)}$ at each node $v \in V$ defines the distributed output of the algorithm. Here $\mathit{B \subseteq \delta(v)}$. $\mathit{steiner\_flag(v)=false}$ ensures that $v$ does not belong to the final ST; in this case $B = \emptyset$. On the other hand $\mathit{steiner\_flag(v)=true}$ ensures that $v$ is a part of the final ST; in this case $B \neq \emptyset$ and for each $e\in B$, $\mathit{ES(e)}$ is set to $\mathit{branch}$.

The special node $r$ initiates the algorithm.  An ordered execution of the steps is necessary for the correct working of the STCCM-A algorithm. We assume that $r$ ensures the ordered execution of the steps (step $1$ to step $4$) and initiates the step $i + 1$ after the step $i$ is terminated. The outline of the proposed STCCM-A algorithm is as follows.

\begin{enumerate}[Step 1.]
\item {\bf SPF-construction}. Construct a SPF $G_F = (V, E_F, w)$ for $Z$ in $G$ by applying the SPF-A algorithm described in Section~\ref{spf-algo-outline}. Theorem~\ref{spf-time-complexity} ensures that the round and message complexities of this step are $\tilde{O}(n^{1/3})$ and $\tilde{O}(n^{7/3})$ respectively.


\item {\bf Edge Weight modification}. With respect to the SPF $G_F$, each edge $e \in E$ of the graph $G= (V, E, w)$ is classified as any one of the following three types. 
\begin{enumerate}
	\item $\mathit{tree}$ edge: if $e \in E_F$.
	\item $\mathit{inter\_tree}$ edge: if $e \notin E_F$ and end points are incident in two different trees of the $G_F$.
	\item $\mathit{intra\_tree}$ edge: if $e \notin E_F$ and end points are incident in the same tree of the $G_F$.
\end{enumerate}

\vspace{.7em}
\noindent
Now transform $G=(V, E, w)$ into $G_c = (V, E, w_c)$. The cost of each edge $(u, v) \in E$ in $G_c$ is computed as follows.
	\begin{enumerate}
		\item $w_c(u, v) = 0$ if $(u, v)$ is a $\mathit{tree}$ edge.
		\item $w_c(u, v) = \infty$ if $(u, v)$ is an $\mathit{intra\_tree}$ edge. 
		\item $w_c(u, v) = d(u, s(u)) + w(u, v) + d(v, s(v))$ if $(u, v)$ is an $\mathit{inter\_tree}$ edge. In this case $w_c(u, v)$ realizes the weight of a path from the source node $s(u)$ to the source node $s(v)$ in $G$ that contains the $\mathit{inter\_tree}$ edge $(u, v)$.  Recall that $d(u, v)$ denotes the (weighted) shortest distance between nodes $u$ and $v$ in $G$.
	\end{enumerate}
	
\smallskip
\noindent
\hspace{1em} The classification of the edges of $G$ and the transformation to $G_c$ can be done as follows. Each node $v$ of $G$ sends a message (say $\langle set\_category(v, s(v), d(v, s(v)), \pi(v)) \rangle$) on all of its incident edges with respect to the input graph $\mathit{G}$. Let a node $v$ receives $\langle set\_category(u, s(u), d(u, s(u)), \pi(u)) \rangle$ on an incident edge $(v, u)$. If $s(v) \neq s(u)$ then $v$ sets $(v, u)$ as an $\mathit{inter\_tree}$ edge and $w_c(v, u)$ to $d(v, s(v)) + w(v, u) + d(u, s(u))$. On the other hand if $s(v) = s(u)$ then $(v, u)$ can be either a $\mathit{tree}$ edge or an $\mathit{intra\_tree}$ edge:  if $v = \pi(u)$ or $\pi(v)=u$ then $v$ sets $(v, u)$ as $\mathit{tree}$ edge and $w_c(v, u)$ to $0$. Otherwise, $v$ sets $(v, u)$ as $\mathit{intra\_tree}$ edge and $w_c(v, u)$ to $\infty$.

\smallskip
\hspace{1em} It is clear that step $2$ can be performed in $O(1)$ rounds. Also on each edge of $G$, the message $\langle set\_category \rangle$ is sent exactly twice (once from each end). Therefore, the message complexity of the step $2$ is $O(m)$.

\item {\bf MST construction}. Construct an MST $\mathit{T_M}$ of $G_c$.  In this step, we apply the deterministic MST algorithm proposed by Lotker et al. \cite{Lotker:2005:MST:1085579.1085591} (a brief description of the algorithm is deferred to ~\ref{Lotker-algo}). To the best of our knowledge, it is the only known deterministic MST algorithm proposed in the CCM till date. All other existing MST algorithms   in the CCM \cite{Hegeman:2015:TOB:2767386.2767434,pemmaraju_et_al:LIPIcs:2016:6882,Ghaffari:2016:MLR:2933057.2933103,Jurdzinski:2018:MOR:3174304.3175472} are randomized in nature. Note that the round and message complexities of the algorithm proposed by Lotker et al. are $O(\log \log n)$ and $O(n^2)$ respectively.

\hspace{.7em} We assume that each node $v \in V$ ($T_M$ contains all the nodes of $V$) knows which of the edges in $\delta(v)$ are in the $\mathit{T_M}$ and for each such edge $\mathit{e}$ it sets $\mathit{ES(e)}$ to $\mathit{branch}$. On the other hand for each $\mathit{e \in \delta(v)}$ which is not in the $\mathit{T_M}$, $v$ sets $\mathit{ES(e)}$ to $\mathit{basic}$.

\item {\bf Pruning}. Construct a generalized MST $T_Z$ by performing a pruning operation on the MST $T_M$. 
For correctness we assume that a node in $Z$ with the smallest ID is the root $r_t$ of $T_M$. The pruning operation deletes edges from $T_M$ until all leaves are terminal nodes.  
It is performed as follows. Each $v \in V$ sends its parent information (with respect to the $T_M$ rooted at $r_t$) to all other nodes. This requires $O(1)$ rounds and $O(n^2)$ messages. Now each $v \in V \setminus Z$ locally computes the $T_M$ rooted at $r_t$ from these received parent information. Then each node in $V \setminus Z$ can locally decide whether it should prune itself or not from the $T_M$. From the locally known $T_M$ each $v \in V \setminus Z$ finds whether it is an intermediate node in between two or more terminals in the $T_M$ or not. If yes, it does not prune itself from the $T_M$ and sets its $\mathit{steiner\_flag}$ to $\mathit{true}$. Otherwise, it prunes itself. Whenever a node $v$ prunes itself from the $\mathit{T_M}$ it sets its $\mathit{steiner\_flag}$ to $\mathit{false}$ and for each $\mathit{e \in \delta(v)}$ such that $\mathit{ES(e) \neq basic}$, it sets its $\mathit{ES(e)}$ to $\mathit{basic}$. On each such pruned edge $e$, $v$ asks the other end of $e$ to prune the common edge $e$. Now the edge weights of resultant ST $T_Z$ are restored to the original edge function $w$. Since each node in $\mathit{V \setminus Z}$ can start the pruning operation in parallel and in the worst-case the number of pruned edges in the network can be at most $n$, this step takes $O(1)$ rounds and $O(n)$ messages. 

The overall round and message complexities of the pruning step are $O(1)$ and $O(n^2)$ respectively.  

\end{enumerate}
The STCCM-A algorithm terminates with the step $4$.

\begin{figure}[ht]
\centering
\vspace{-14em}
\hspace*{-4em}
  \begin{subfigure}[b]{0.5\textwidth}
    \includegraphics[width=10cm, height=14cm]{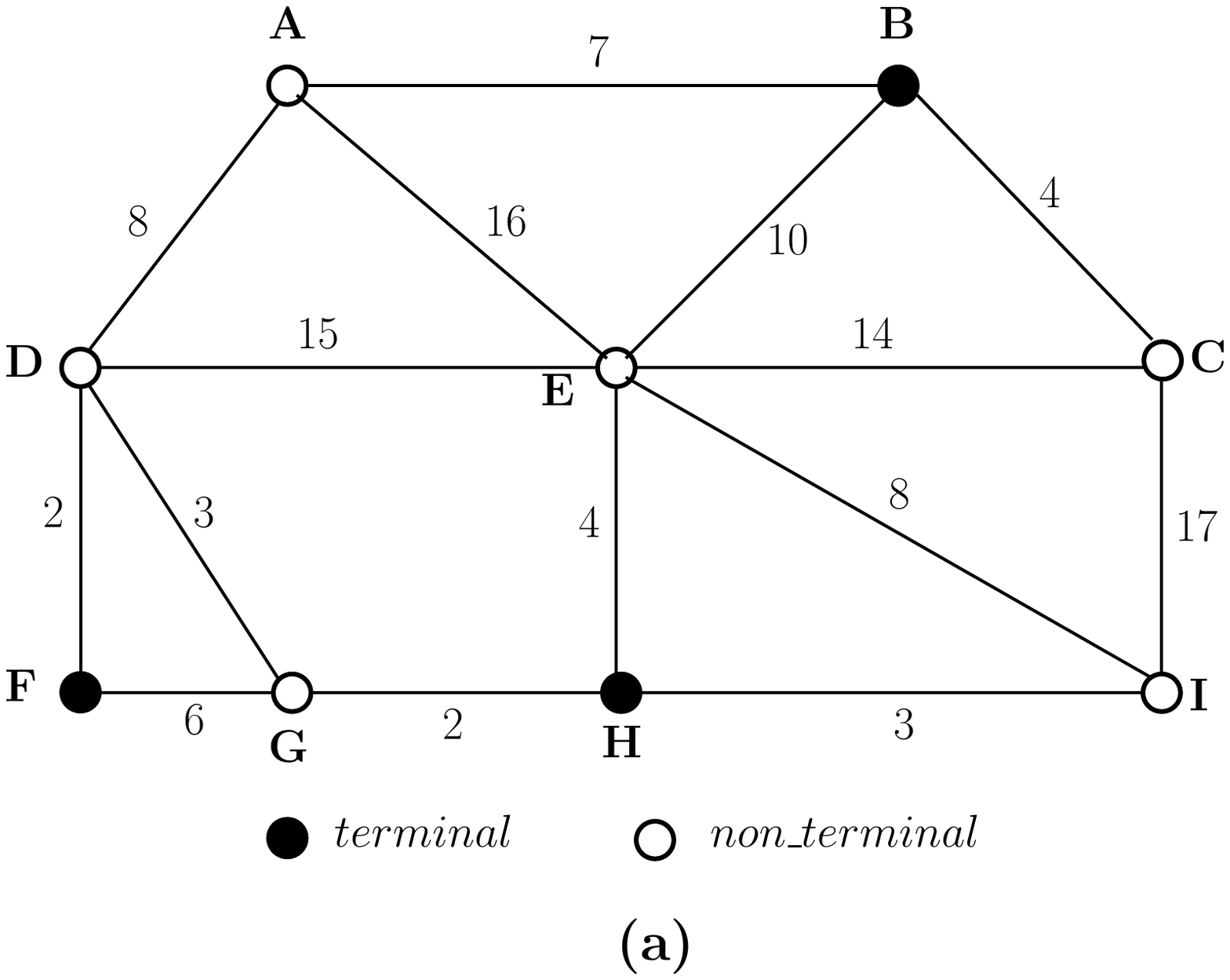}
    \vspace{-11em}
    \centering
  \end{subfigure}
  \begin{subfigure}[b]{0.5\textwidth}
    \includegraphics[width=10.5cm, height=13cm]{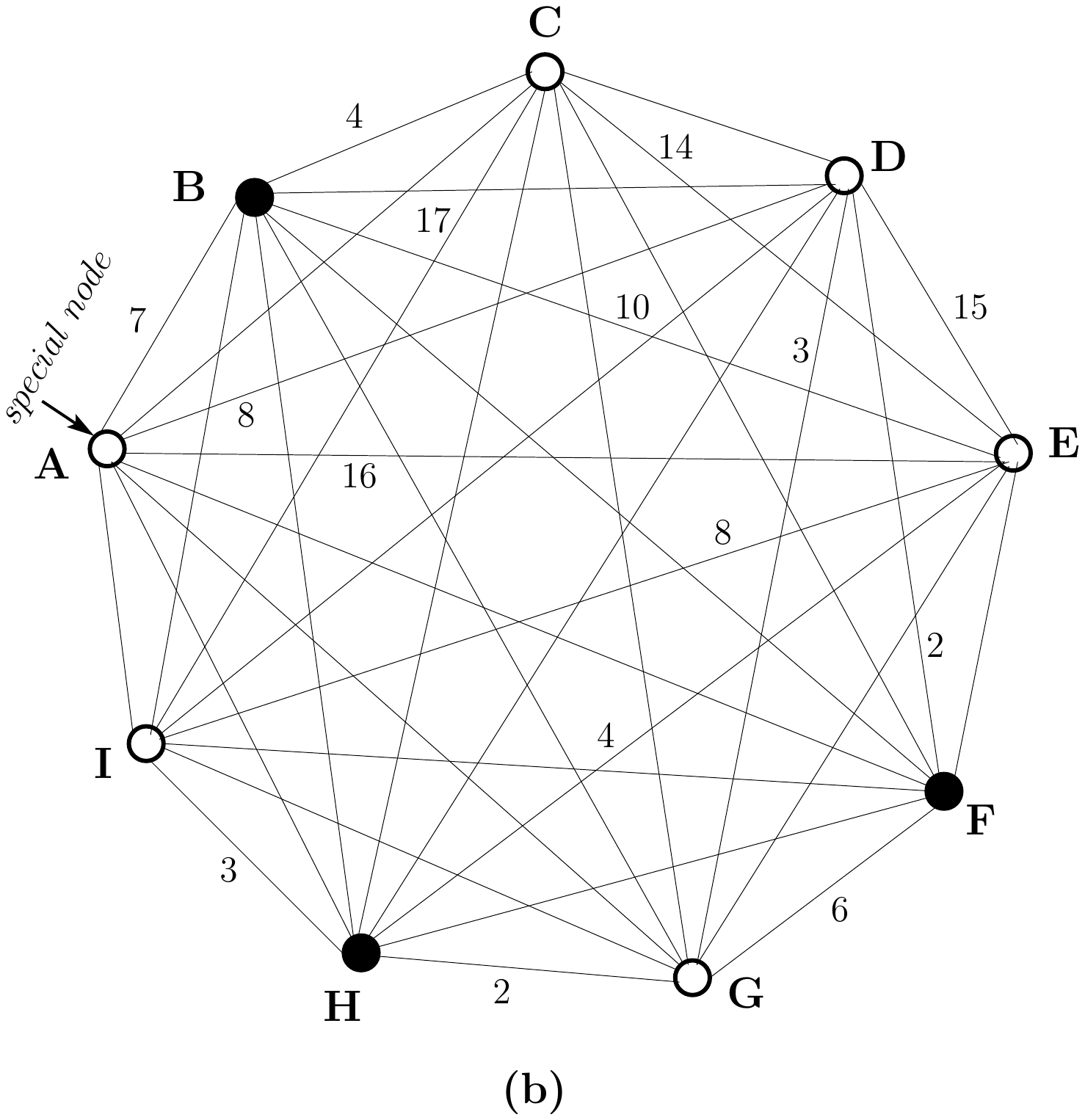}
    \vspace{-10em}
  \end{subfigure}
\end{figure}

\vspace{7em}

\begin{figure}[ht]
\centering
\vspace{-5em}
\hspace*{-4em}
  \begin{subfigure}[b]{0.5\textwidth}
    \includegraphics[width=10.5cm, height=13cm]{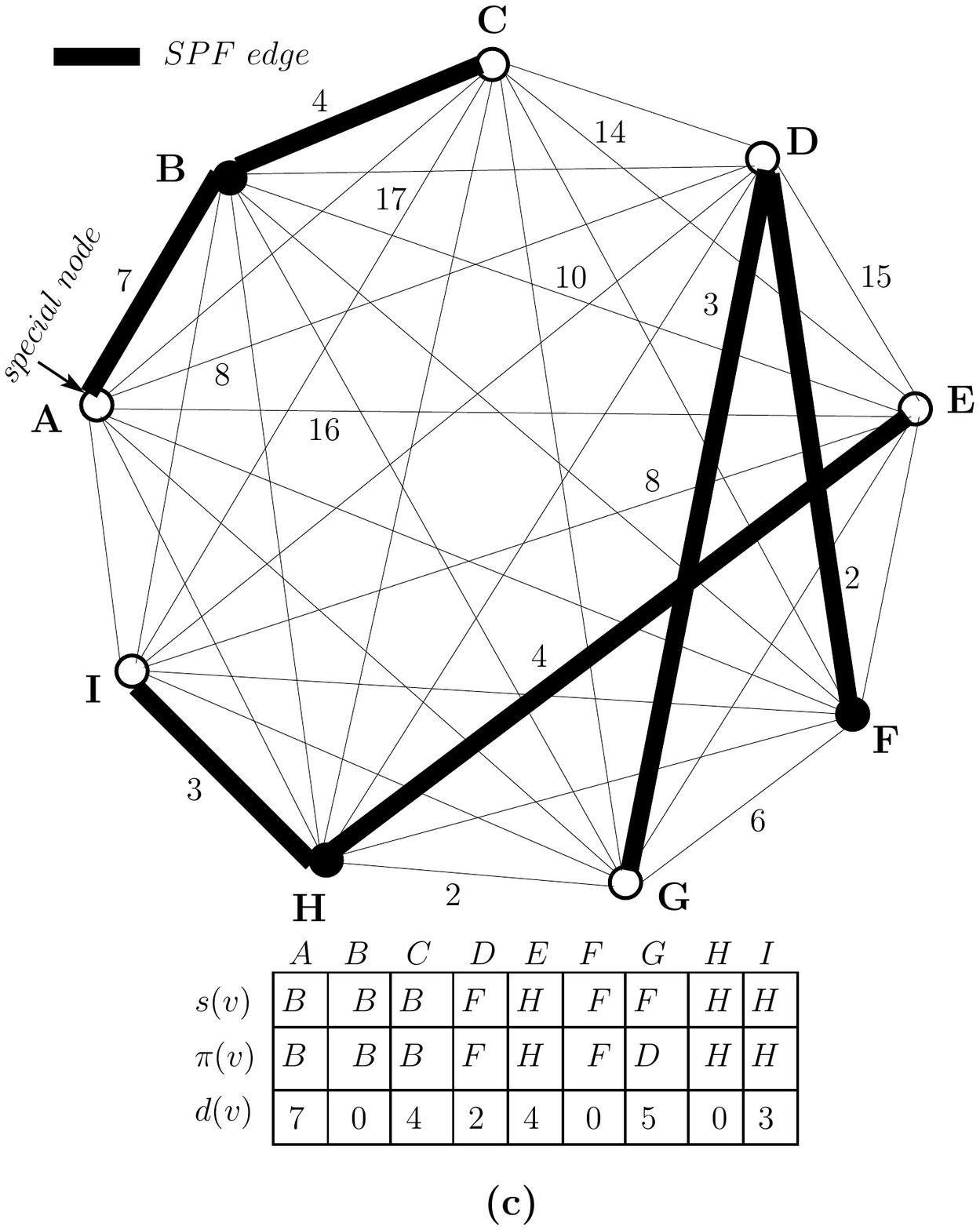}
    \vspace{-7em}
    \centering
  \end{subfigure}
  \begin{subfigure}[b]{0.5\textwidth}
  \includegraphics[width=10.5cm, height=13cm]{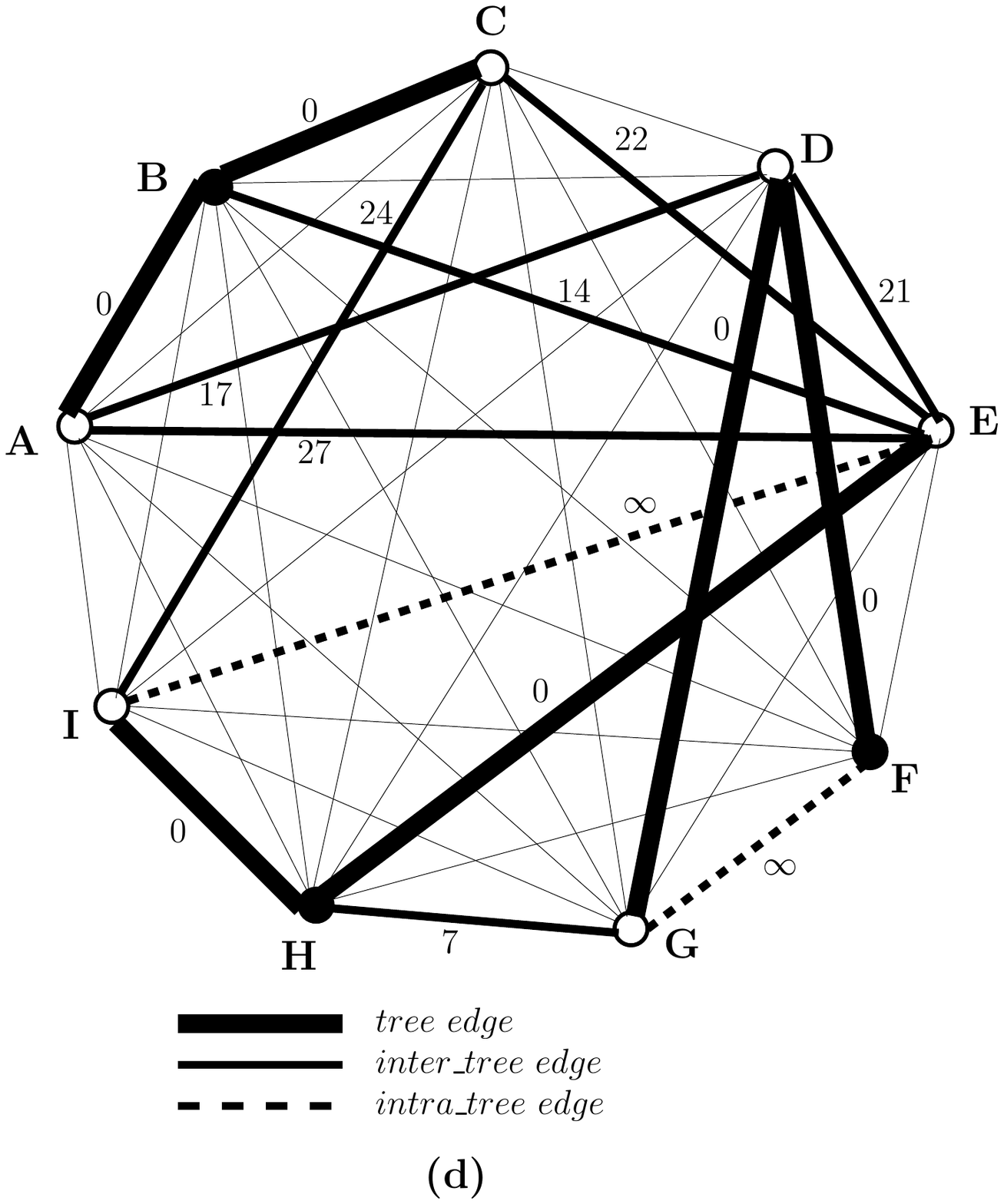}
    \vspace{-7em}
  \end{subfigure}
 \caption{(a) An arbitrary graph $G = (V, E, w)$ and a terminal set $Z = \{B, F, H\}$. (b) The input graph $G$ is distributed among the processors of a complete network $K_9$ via a vertex partition. (c) A SPF $G_F =(V, E_F, w)$ for $Z$ of $G$. The distances of nodes to their respective sources are shown in the table. (d) The modified graph $G_c = (V, E, w_c)$. (e) An MST $T_M$ of $G_c$. (f) The final Steiner tree $T_Z$ (generalized MST) for $Z$ of $G$.}  \label{fig:stccm-fig-example}
\end{figure}

\begin{figure}[ht]
\centering
\vspace{-9em}
\hspace*{-4em}
  \begin{subfigure}[b]{0.5\textwidth}
    \includegraphics[width=10.5cm, height=13cm]{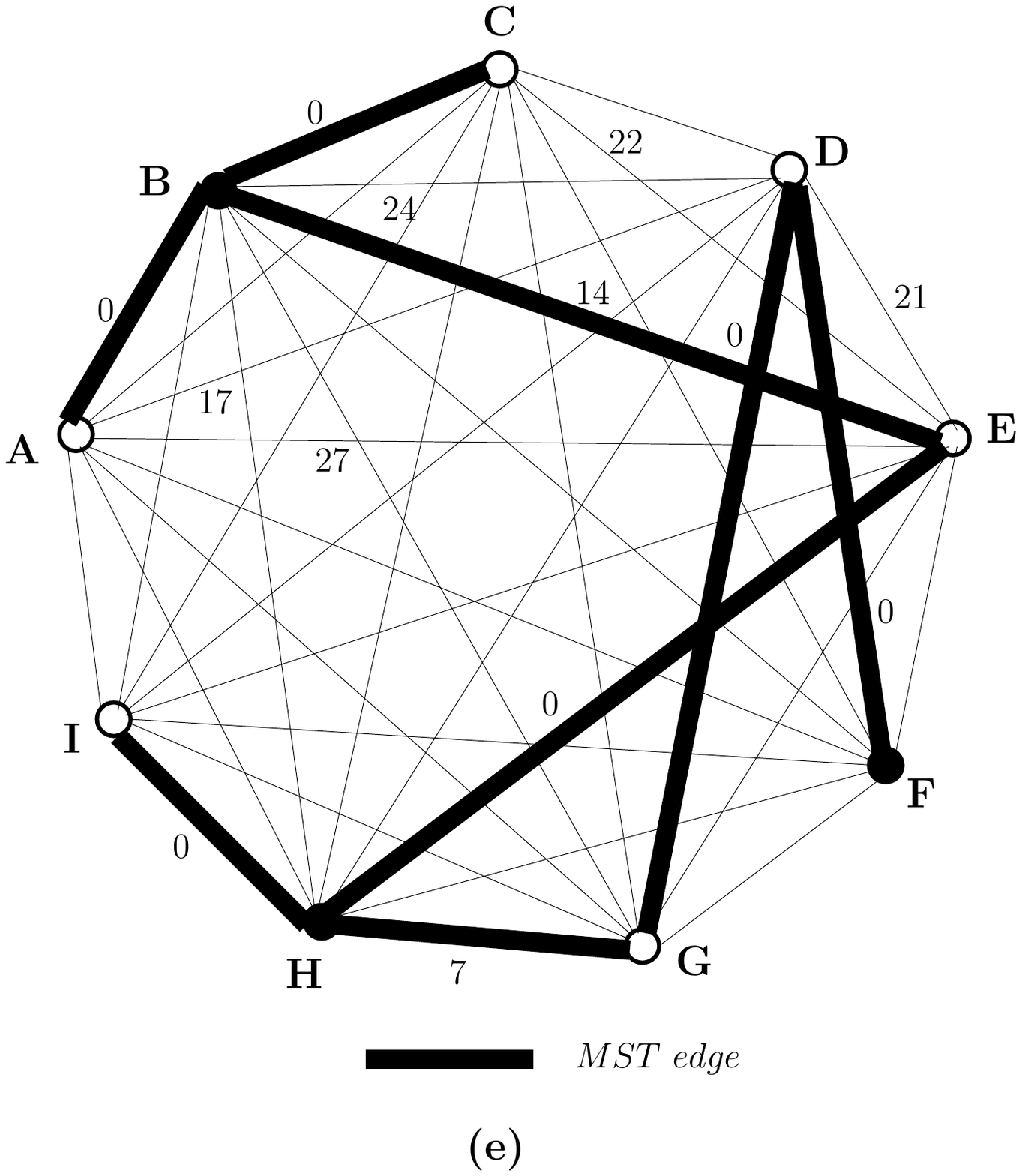}
    \vspace{-10.5em}
    \centering
  \end{subfigure}
  \begin{subfigure}[b]{0.5\textwidth}
    \includegraphics[width=10.5cm, height=13cm]{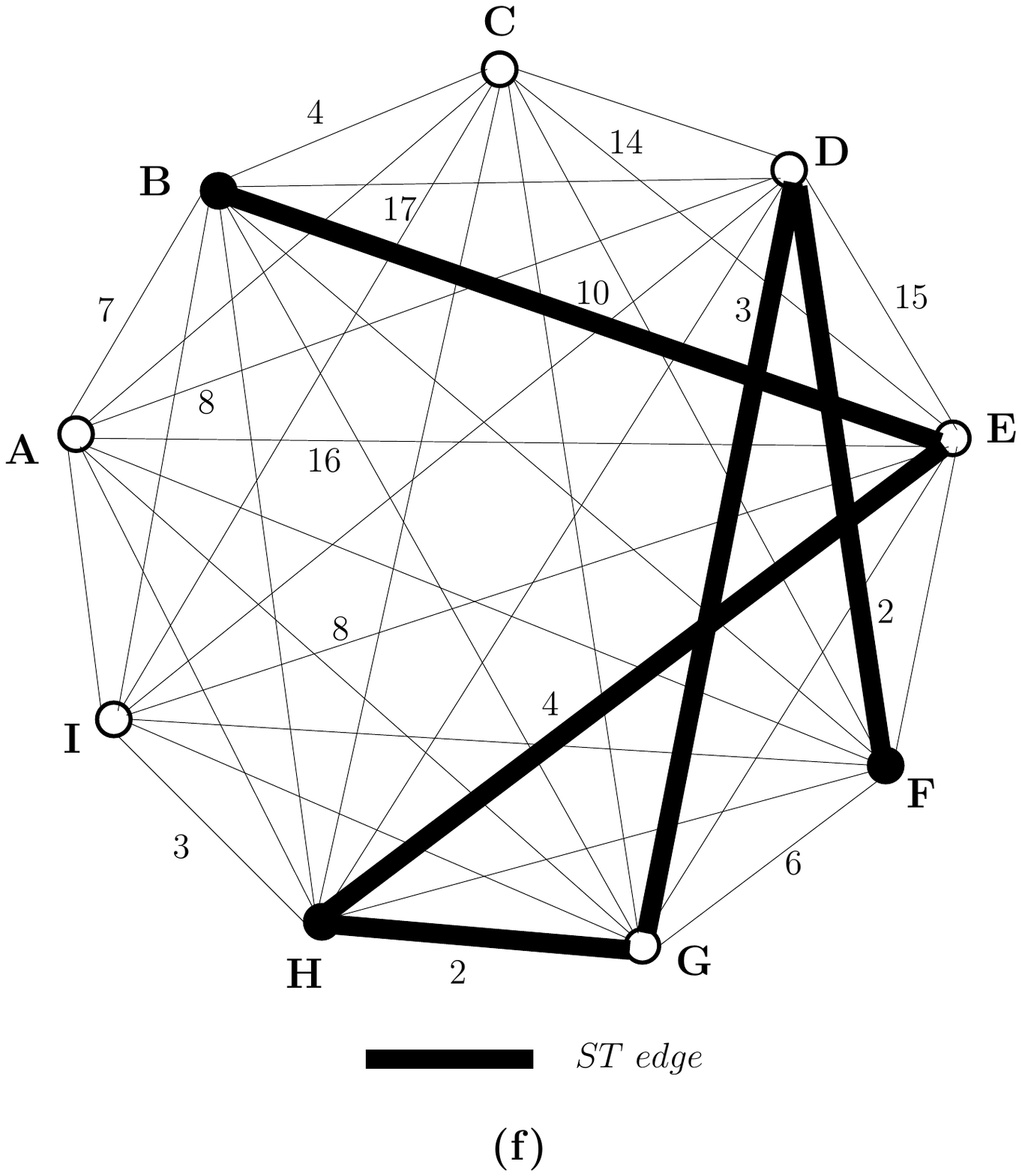}
    \vspace{-10em}
  \end{subfigure}
  \vspace{.2em}
  \caption*{Figure~\ref{fig:stccm-fig-example} (continued).} \label{fig:stccm-ccm-ex-fig-5}
\end{figure}

\subsection{An illustrating example of the STCCM-A algorithm}
Consider the application of the STCCM-A algorithm in an arbitrary graph $G = (V, E, w)$ and a terminal set $Z = \{B, F, H\}$ shown in Figure~\ref{fig:stccm-fig-example}(a). The input graph $G$ is distributed among the nodes (processors) of a complete network $K_9$ via a vertex partition which is shown in Figure~\ref{fig:stccm-fig-example}(b). Specifically each vertex and its incident edges (with weights) of $G$ are assigned to a distinct processor in the $K_9$. The weight of an edge in $K_9$ which is not in $G$ is considered equal to $\infty$ (not shown in the figure). A SPF $G_F =(V, E_F, w)$ for $Z$ is constructed which is shown in Figure~\ref{fig:stccm-fig-example}(c). In $G_F$, each non-terminal $v$ is connected to a terminal $s(v) \in Z$ whose distance is minimum to $s(v)$ than any other terminal in $G$ which is shown in the table of Figure~\ref{fig:stccm-fig-example}(c). The modified graph $G_c$ and labelling of the edge weights according to the definition of $G_c$ are shown in Figure~\ref{fig:stccm-fig-example}(d). Figure~\ref{fig:stccm-fig-example}(e) shows after the application of the Lotker et al.'s MST algorithm on $G_c$ which constructs an MST $T_M$ of $G_c$.  The final Steiner tree $T_Z$ for $Z$ of $G$, which is a generalized MST for $Z$ of $G$ is constructed from the $T_M$ by applying the pruning step of the STCCM-A algorithm, which is shown in Figure~\ref{fig:stccm-fig-example}(f).

\vspace{1em}
\section{Proof of the STCCM-A algorithm} \label{proof-stccm-algorithm}
\begin{theorem} \label{stc-time-proof}
The round complexity of the STCCM-A algorithm is $\tilde{O}(n^{1/3})$.
\end{theorem}
\begin{proof}
It is clear that the overall round complexity of the STCCM-A algorithm is dominated by the step $1$. Therefore the round complexity of the STCCM-A algorithm is $\tilde{O}(n^{1/3})$. The polylogarthmic factors involved with this round complexity is $\log n$.
\end{proof}

\begin{theorem} \label{stc-mesage-proof}
The message complexity of the STCCM-A algorithm is $\tilde{O}(n^{7/3})$.
\end{theorem}
\begin{proof}
By Theorem~\ref{spf-time-complexity} the message complexity of the step 1 of the STCCM-A algorithm is $\tilde{O}(n^{7/3})$. Each of the step 3 and step 4 requires $O(n^2)$ messages. The step $2$ requires $O(m)$ messages. We know that $m \leq n^2$. All of these ensures that the overall message complexity of the STCCM-A algorithm is dominated by the step $1$. Therefore the message complexity of the STCCM-A algorithm is $\tilde{O}(n^{7/3})$. The polylogarthmic factors involved with this message complexity is $\log n$.
\end{proof}


\begin{definition}[{\bf Straight path}]
Given that $G = (V, E, w)$ is a connected undirected weighted  graph and $Z \subseteq V$. Let $u, v \in Z$. Then a path $P_{uv}$ between $u$ and $v$ is called straight only if all the intermediate nodes in $P_{uv}$ (may contain no intermediate node) are in $V \setminus Z$.  
\end{definition}

\begin{lemma}
Given that $G = (V, E, w)$ is a connected undirected weighted  graph and $Z \subseteq V$. Let $u, v \in \mathit{Z}$ and there exists a straight path $P_{uv}$ between $u$ and $v$ in $\mathit{T_M}$, where $T_M$ is a resultant MST constructed  after the consecutive applications of step $1$, $2$, and $3$ of the STCCM-A algorithm on the given  graph $G$. Then $P_{uv}$ in $T_M$ is the shortest straight path between $u$ and $v$ in $G$.  
\end{lemma}

\begin{proof}
By contradiction assume that $P_{uv}$ is not the shortest straight path in $T_M$. Let there exists a straight path $P'_{uv}$ between $u$ and $v$ such that $P'_{uv} < P_{uv}$. We show that $P'_{uv} < P_{uv}$ does not hold. Since $u$ and $v$ are two different terminals they were in different trees of the SPF, say $u \in \mathit{T_1}$ and $v \in \mathit{T_2}$, before they are included in $\mathit{T = T_1 \cup T_2}$. Let $\mathit{e = (a, b) \in \delta(T_1) \cap \delta(T_2), e' = (a', b') \in \delta(T_1) \cap \delta(T_2)}$ such that $P_{uv}$ and $P'_{uv}$ contain $e$ and $e'$ respectively as shown in Figure~\ref{fig:straight-path}. Note that the correctness of the SPF algorithm described in Section~\ref{spf-algo-outline} ensures that $P_{ua}$, $P_{ua'}$, $P_{vb}$, and $P_{vb'}$ are the shortest paths between the respective nodes in $G$ as shown in Figure~\ref{fig:straight-path}. During the execution of the step $3$ of the STCCM-A algorithm, which constructs the MST $T_M$, $a$ finds $e$ as its minimum weight outgoing edge (MWOE) and $a'$ finds $e'$ as its MWOE. Similarly $b$ and $b'$ finds $e$ and $e'$ as their MWOEs respectively. During the construction of the MST $T_M$ since $\mathit{T_1}$ and $\mathit{T_2}$ are merged through $e$, this indicates that $e$ is the MWOE of $\mathit{T_1}$. Similarly $e$ is the MWOE of $\mathit{T_2}$. Since $\mathit{T_1}$ and $\mathit{T_2}$ are merged along the path $P_{uv}$ to form $\mathit{T}$, this ensures that $P'_{uv} \geq P_{uv}$ contradicting the fact  $P'_{uv} < P_{uv}$. Therefore $P_{uv}$ is the shortest straight path between $u$ and $v$ in $G$. 

\begin{figure}[ht!]
	\vspace{-10em}
	\centering
    \includegraphics[width=12cm, height=16cm]{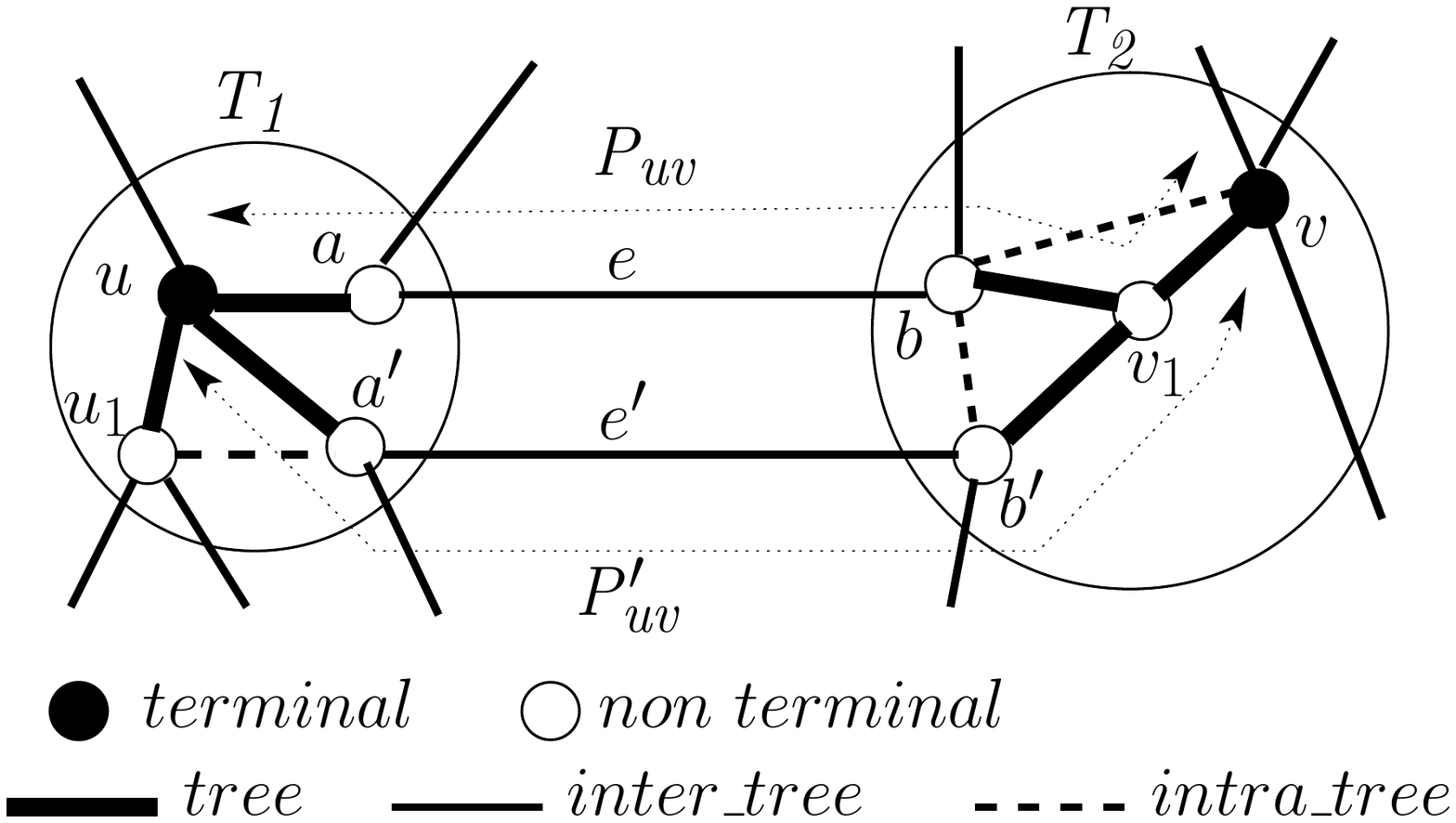}
    \vspace{-14em}
    \centering
    \caption{A state before merging of two shortest path trees $\mathit{T_1}$ and $\mathit{T_2}$ along a shortest straight path $\mathit{P_{uv}}$. The edge categories in $\mathit{T_1}$ and $\mathit{T_2}$ are shown according to the graph $\mathit{G_c}$ constructed in step $2$ of the STCCM-A algorithm. Note that both the trees $\mathit{T_1}$ and $\mathit{T_2}$ are subgraphs of a complete graph $K_n$.}\label{fig:straight-path}
\end{figure}
\end{proof}

Consider a graph $T^* = (Z, E^*, w^*)$ whose vertex set $Z$ and edge set $E^*$ are defined from  $T_M$ as follows. For each \emph{straight path} $P_{uv}$ of $T_M$, let $(a, b)$ be an edge of $E^*$. Then the following lemma holds.

\begin{lemma} \label{k-mst}
$T^*$ is an MST of $K_Z$ for $Z \subseteq V$ of graph $G = (V, E, w)$.
\end{lemma}

The correctness of the above lemma directly follows from the correctness of the Lemma~$4.8$ in \cite{Saikia-Karmakar2019}.

\medskip
\noindent
It is obvious that if we unfold the tree $T^* = (Z, E^*, w^*)$ then it transforms to a resultant graph $\mathit{T_Z}$, which satisfies the following properties.
\begin{itemize}
\item For each straight path between $u$ and $v$ in $\mathit{T_Z}$, there exists an edge $e = (u, v) \in E^*$, where $u, v \in Z$ and the length of the straight path between $u$ and $v$ in $\mathit{T_Z}$ is the shortest one between $u$ and $v$ in $\mathit{G}$.
\item All leaves of $\mathit{T_Z}$ are in $\mathit{Z}$.  

\end{itemize} 

Therefore the following theorem holds.

\begin{theorem} \label{g-mst}
The tree $T_Z$ computed by the STCCM-A algorithm is a generalized MST for $Z \subseteq V$ of $G = (V, E, w)$.
\end{theorem}

Let $cost(X)$ denotes the sum of weights of all the edges in a graph $X$. Let $T_{opt}$ denotes the optimal ST. Then the following theorem holds.

\begin{theorem}
$cost(T_Z) /cost(T_{opt}) \leq 2(1 - \frac{1}{\ell}) \leq 2(1 - \frac{1}{|Z|})$.
\end{theorem}
The correctness of the above theorem essentially follows from the correctness of the Theorem 1 of Kou et al. \cite{Kou:1981:FAS:2697742.2698014}. For the sake of completeness here we give the outline of this correctness. Let $T_{opt}$ consists of $q \geq 1$ edges. Then there exists a loop $L = v_0, v_1, v_2, \cdots, v_{2q}(=v_0)$ in $T_{opt}$ in such a way that
\begin{itemize}
	\item every edge in $T_{opt}$ appears exactly twice in $L$. This implies that $cost(L) = 2 \times cost(T_{opt})$.   
	\item every leaf in $T_{opt}$ appears exactly once in $L$ and if $v_i$ and $v_j$ are two consecutive leaves in $L$ then the sub-path connecting $v_i$ and $v_j$ is a simple path. 
\end{itemize} 
Note that the loop $L$ can be decomposed into $\ell$ simple sub-paths (recall that $\ell$ is the number of leaf nodes in the $T_{opt}$), each of them connects two consecutive leaf nodes in $L$. By deleting the longest such simple sub-path from $L$ the remaining path (say $P$) in $L$  satisfies the followings.
\begin{itemize}
	\item every edge in $T_{opt}$ appears at least once in $P$.
	\item $cost(P) \leq (1 - \frac{1}{\ell}) \times cost(L) = (1 - \frac{1}{\ell}) \times 2 \times cost(T_{opt})$. 
\end{itemize}


Now assume that $T_Z$ is a generalized MST for $Z$ of $G$. We know that $T_Z$ realizes the MST say $T_{A}$ of the complete distance graph $K_Z$ for $Z$ of $G$. In other words $cost(T_Z) = cost(T_A)$. Note that each edge $(u, v)$ in $T_A$, where $u, v \in Z$, is a \emph{shortest straight path} between $u$ and $v$ in $G$. This ensures that the weight of an edge $(u, v)$ in $T_A$ is at most the weight of the corresponding simple path between nodes $u$ and $v$ in $P$.  If we consider all the edges of $T_A$ then $cost(T_A) \leq cost(P)$. This concludes that $cost(T_Z) = cost(T_A) \leq cost(P) \leq 2 \times (1 - \frac{1}{\ell}) \times cost(T_{opt})$. 

\bigskip
\noindent
{\bf Deadlock issue.} The STCCM-A algorithm is free from deadlock. Deadlock occurs only if a set of nodes in the system enter into a circular wait state. In step $1$ (SPF construction) of the STCCM-A algorithm, a node uniformly distributes its input (a brief description of the input distribution is given in ~\ref{Censor-Hillel-all-algo}), which is the incident edge weights of it, among a subset of nodes (which are of course withing one hop distances) and never sends resource request message to any other nodes. This ensures that nodes never create circular waiting (or waiting request) for any resource during the execution of the algorithm. In Step $2$, each node independently sends a message (containing its own state) to all of its neighbors only once. Upon receiving messages from all of its neighbors, a node performs some local computation and then terminates itself.  Therefore, in step $2$ nodes are free from any possible circular waiting. The correctness of the deadlock freeness of the step 3 essentially follows from the work of Lotker et al. \cite{Lotker:2005:MST:1085579.1085591}. In step $4$, the pruning operation is performed on a tree structure ($T_M$) and a node never requests and waits for resources holding by any other node in the system. This implies that during the pruning operation, nodes are free from any possible circular waiting. Therefore, the deadlock freeness of all four steps together ensure deadlock freeness of the STCCM-A algorithm.

\section{STCCM-B algorithm} \label{stccm-b-algorithm}
The STCCM-B algorithm is a modified version of the STCCM-A algorithm that computes an ST in $O(S + \log \log n)$ rounds and $O(S(n - t)^2   + n^2)$ messages in the CCM maintaining the same approximation factor as that of the STCCM-A algorithm described in Section~\ref{stccm-algorithm}. Similar to the STCCM-A algorithm the STCCM-B algorithm also has four steps (small distributed algorithms). In the STCCM-B algorithm except the step $1$ all other steps are same as that of the STCCM-A algorithm. Recall that step $1$ of the STCCM-A algorithm is a construction of a SPF of a given input graph $G = (V, E, w)$ with a terminal set $Z \subseteq V$. In contrast for step $1$ of the STCCM-B algorithm, which is also the SPF construction, we adapt the SPF algorithm proposed in \cite{Saikia-Karmakar2019} to the CCM. This new SPF algorithm denoted by SPF-B helps achieving a different round and message complexities for $2(1 - 1/\ell)$-approximate ST construction in the CCM. Specifically for constant or small shortest path diameter networks (where $S = O(\log \log n)$) the STCCM-B algorithm outperforms the STCCM-A algorithm in terms of round complexity in the CCM.

\medskip
\noindent
{\bf The SPF algorithm in \cite{Saikia-Karmakar2019} vs. the SPF-B algorithm.} The SPF algorithm in \cite{Saikia-Karmakar2019} was proposed for the $\mathcal{CONGEST}$ model, whereas our SPF-B algorithm is devised for the CCM. In the SPF algorithm in \cite{Saikia-Karmakar2019} the terminal nodes need to participate in forwarding $\langle update \rangle$ messages until the termination of the algorithm. However in the SPF-B algorithm the terminal nodes participate in forwarding $\langle update \rangle$ messages only once and after that they exchange no further messages in the system. The SPF-B algorithm terminates in $S + 2$ rounds, whereas the SPF algorithm in \cite{Saikia-Karmakar2019} terminates in $S + h$ rounds, where $h$ is the height of a breadth first search tree of the given input graph $G$.  Note here that $h \leq D \leq S$. Despite the same asymptotic round complexities, required by the both algorithms, which is $O(S)$, they incur different message complexities. The SPF-B algorithm incurs a message complexity of $O(S(n - t)^2 + nt)$, where $n > t$, whereas the SPF algorithm in \cite{Saikia-Karmakar2019} has the message complexity of $O(\Delta (n - t) S)$, where $\Delta$ is the maximum degree of a vertex in the input graph $G$.

\bigskip
\subsection{SPF-B algorithm}
In addition to the notations defined in Section~\ref{model-notations} the following notations are used in the description of the SPF-B algorithm.
\begin{itemize}
\item $ts(v)$ denotes the {\em tentative source ID} of a node $v$. 
	\item $td(v)$ denotes {\em tentative shortest distance} from a node $v$ to its $ts(v)$. 
	\item $tf(v)$ denotes the {\em terminal flag} of a node $v$.
	\item $t\pi(v)$ denotes the {\em tentative predecessor} of a node $v$. 
	\item Let $e \in \delta(v)$. Then at node $v$, $tdn(e)$ denotes the {\em tentative shortest distance} of a node incident on the other end of $e$. Similarly at node $v$, $idn(e)$, $tsn(e)$ and $tfn(e)$ denote {\em ID}, {\em tentative source ID} and {\em terminal flag} value respectively of a node incident on the other end of $e$.

\end{itemize}

\medskip
\noindent
{\bf Input specification.}  
\begin{itemize}
\item If $\mathit{v \in Z}$ then $td(v) = 0$, $t\pi(v) = v$, $ts(v) = v$, and  $\mathit{tf(v) = true}$.
\item If $\mathit{v \in V \setminus Z}$ then $td(v) = \infty$, $t\pi(v) = nill$, $ts(v) = nill$, and $\mathit{tf(v) = false}$.
\item The value of $\mathit{ES(e)}$ can be either $\mathit{basic}$ or $\mathit{blocked}$. A node $v$ sets the state of an incident edge $(v, u) \in \delta(v)$ as $\mathit{blocked}$ if $u \in Z$. Otherwise, at node $v$, the edge $(v, u)$ has the state value as $\mathit{basic}$. It is possible for the edge states at the two nodes adjacent to the edge to be temporarily different.
\end{itemize}

\medskip
\noindent
{\bf Output specification.} When the algorithm terminates then $d(v) = td(v)$, $\pi(v) = t\pi(v)$, and $s(v) = ts(v)$ for each node $v \in V$.


\begin{algorithm}[b!]
\caption*{{\bf SPF-B algorithm.} Pseudocode at node $v$ upon receiving a set of messages or no message.} \label{alg:SPF-B}
\begin{flushleft} Let $Z$ be the set of terminals, $U$ be the set of $\langle update \rangle$ messages received by node $v$ in a round.\end{flushleft} 
\begin{algorithmic}[1]
\State {\bf upon receiving no message}		
\If {$v = r$}		\Comment{spontaneous awaken of the special node $r$}
	\ForAll {$\mathit{e \in \delta(v)}$} \Comment{here $\mathit{\delta(v)}$ contains $n - 1$ edges}
		\State {{\bf send} $\langle wakeup \rangle$ on $\mathit{e}$}
	\EndFor
\EndIf

\vspace{2em}
\State {\bf upon receiving $\langle wakeup \rangle$}
\ForAll {$\mathit{e \in \delta(v)}$}
		\State {$\mathit{ES(e) \leftarrow basic;}$}
\EndFor
\If {$\mathit{v \in Z}$}
	\State {$tf \leftarrow true; ts \leftarrow v; t\pi \leftarrow v;td \leftarrow 0;$}
	\ForAll {$\mathit{e \in \delta(v)}$} 
		\State {{\bf send} $\langle update(v, ts, td, tf) \rangle$ on $\mathit{e}$} 
	\EndFor 
\Else
	\State {$tf \leftarrow false;  ts \leftarrow nill; t\pi \leftarrow nill; td \leftarrow \infty;$}
	
\EndIf

\vspace{2em}
\State {\bf upon receiving a set of $\langle update \rangle$ messages} \Comment{$U \neq \phi$}
\ForAll {$\langle update(idn(e), tsn(e), tdn(e), tfn(e)) \rangle \in U$ such that $e \in \delta(v)$ } 
	
		\If {$tfn(e) = true$}
			\State {$ES(e) \leftarrow blocked;$}		
		\EndIf

		\If {$tdn(e) + w(e) < td$} \Comment{update the tentative source, distance, and predecessor}
			\State{$update\_flag \leftarrow true;$} \Comment{$update\_flag$ is a temporary boolean variable}
			\State {$td \leftarrow tdn(e) + w(e); t\pi \leftarrow idn(e); ts \leftarrow tsn(e);$}
		\EndIf
	\EndFor
	\If{$update\_flag = true$}	
		\ForAll {$e \in \delta(v)$ such that $ES(e) \neq blocked$} 
			\State {{\bf send} $\langle update(v, ts, td, tf) \rangle$ on $e$} 
		\EndFor
		\State{$update\_flag \leftarrow false;$}
	\EndIf

\end{algorithmic}
\end{algorithm}

\medskip
\noindent
{\bf Outline of the algorithm.} We assume that the SPF-B algorithm is initiated by a special node $r$. $r$ initiates the algorithm by sending a $\langle wakeup \rangle$ message to all other nodes in the communication network, which is a clique. Upon receiving the $\langle wakeup \rangle$ message, a node $v$ initializes its local variables as shown in the input specification. Now each node $\mathit{v \in Z}$ sends $\langle update \rangle$ messages to all nodes in $V$. After that nodes in $\mathit{Z}$ send or receive no further messages, whereas in each subsequent rounds nodes in $\mathit{V \setminus Z}$ may send or receive $\langle update \rangle$ messages until the termination of the algorithm. Upon receiving a set of $\langle update \rangle$ messages (denoted by $U$) a node $v \in V$ acts as per the following rules.  

\begin{enumerate}[R1.]
	\item if an $\langle update \rangle \in U$ is received through an incident edge $e = (v, u)$ such that $u \in Z$ then it sets $\mathit{ES(e)}$ to $\mathit{blocked}$. 
	\item if $v \in V \setminus Z$ then it computes $w(e) + tdn(e)$ for each $\langle update(idn(e), tsn(e), tdn(e), tfn(e)) \rangle \in U$ and chooses the minimum one, say $w(e') + tdn(e')$ resulted by $\langle update(idn(e'), tsn(e'),  tdn(e'), $ \\ $ tfn(e')) \rangle \in U$. If $td(v) > w(e') + tdn(e')$ then it updates $td(v) = w(e') + tdn(e')$, $ts(v) = tsn(e')$, and $t\pi(v) = idn(e')$.  Otherwise, $td(v)$, $ts(v)$ and $t\pi(v)$ remain unchanged.
	\item if the value of $td(v)$ is updated then it sends $\langle update(v, ts(v), td(v), tf(v)) \rangle$ on all of its incident $\mathit{basic}$ edges. 
\end{enumerate}

\begin{lemma} \label{lemma-spfccm-termination}
The SPF-B algorithm terminates after at most $S + 2$ rounds.
\end{lemma}
\begin{proof}
The special node $r$ initiates the SPF-B algorithm by sending $\langle wakeup \rangle$ messages to all the nodes in $V$ which takes exactly $1$ round. Upon receiving $\langle wakeup \rangle$ message each node in $\mathit{Z}$, in parallel, sends $\langle update \rangle$ messages to each node in $V$. This takes $1$ round only. After that all nodes in $\mathit{V \setminus Z}$ proceed in parallel, and if applicable they update their local states. Note that in each subsequent round at least one node in $\mathit{V \setminus Z}$ must update its local variables, otherwise, the algorithm must have reached the termination state. Since $S$ is the shortest path diameter any path in the SPF contains no more than $S$ edges. Since all  nodes in $\mathit{V \setminus Z}$ update their local states in parallel, each such  node $v$ converges to its correct $d(v)$ value in at most $S$ additional rounds. Therefore after $S + 2$ rounds of execution no local changes occur at any node in the system and consequently no further messages related to the SPF-B algorithm will be sent or in transit. This concludes that the SPF-B algorithm terminates after at most $S + 2$ rounds. 
\end{proof}

\begin{theorem}
\label{spfccm-time-complexity}
The round complexity of the SPF-B algorithm is $O(S)$.
\end{theorem}
\begin{proof}
By Lemma~\ref{lemma-spfccm-termination} the SPF-B algorithm terminates after at most $S + 2$ rounds. Therefore the overall round complexity of the SPF-B algorithm is $O(S)$. 
\end{proof}

\begin{theorem}
\label{spf-message-complexity}
The message complexity of the SPF-B algorithm is $O(S(n - t)^2 + nt)$, where $n > t$.
\end{theorem}
\begin{proof}
The special node $r$ initiates the SPF-B algorithm by sending $\langle wakeup \rangle$ messages to all the nodes in $V$, which generates exactly $n$ messages (we assume that a node can send message to itself too).  Upon receiving $\langle wakeup \rangle$ message each node in $\mathit{Z}$, in parallel, sends $\langle update \rangle$ messages to each node in $V$. This step generates $nt$ messages, where $t = |Z|$. After that no further messages are generated due to the node set $\mathit{Z}$. On the other hand in each subsequent round each node in $\mathit{V \setminus Z}$ may update its local variables and sends $\langle update \rangle$ message to each node in $\mathit{V \setminus Z}$. In the worst-case in each subsequent round this may generates $(n - t)^2$ messages. By Lemma~\ref{lemma-spfccm-termination} the SPF-B algorithm terminates in $S + 2$ rounds. Combining all the steps  the total number of messages generated, in the worst-case, is $(S + 2)(n - t)^2 + nt + n$. We assume that $n > t$. Therefore the overall message complexity of the SPF-B algorithm is $O(S(n - t)^2 + nt)$.  
\end{proof}

\begin{lemma}\label{correctness-spf}
Let $td_i(v)$ be the length of the tentative shortest path from node $v$ to its $ts(v)$ after $i \geq 0$ rounds. Let the SPF-B algorithm terminates after $X \leq S + 2$ rounds. Then $td_{X}(v) = d(v, s(v))$ for each node $v \in V$ that uses $\leq S$ edges.
\end{lemma}

The correctness of the above lemma directly follows from the correctness of the Lemma $3.3$ in \cite{Saikia-Karmakar2019}.

\subsection{Proof of the STCCM-B algorithm}
Similar to the STCCM-A algorithm the STCCM-B algorithm also computes a generalized MST for $\mathit{Z}$ of $\mathit{G}$. Therefore the approximation factor of an ST computed by the STCCM-B algorithm is same as that of the STCCM-A algorithm, which is $2(1 - 1/\ell)$.

\begin{theorem} \label{m-stccm-round-complexity}
The round complexity of the STCCM-B algorithm is $O(S + \log \log n)$.
\end{theorem}
\begin{proof}
We know that the STCCM-B algorithm consists of four steps. By Theorem~\ref{spfccm-time-complexity} the round complexity of the step $1$ is $O(S)$. The round complexities of the steps $2$, $3$, and $4$ of the STCCM-B algorithm are $O(1)$, $O(\log \log n)$, and $O(1)$ respectively. It is clear that the overall round complexity of the STCCM-B algorithm is dominated by the steps $1$ and $3$, which is $O(S + \log \log n)$.
\end{proof}

\begin{theorem} \label{m-stccm-message-complexity}
The message complexity of the STCCM-B algorithm is $O(S (n - t)^2 + n^2)$.
\end{theorem}
\begin{proof}
By Theorem~\ref{spf-message-complexity} the message complexity of the step 1 of the STCCM-B algorithm is $O(S (n - t)^2 + nt)$. The message complexities of the steps $2$, $3$, and $4$ of the STCCM-B algorithm are $O(m)$, $O(n^2)$, and $O(n^2)$ respectively. We know that $m \leq n^2$ and $nt < n^2$, where $n > t$. Therefore the overall message complexity of the STCCM-B algorithm is $O(S (n - t)^2 + n^2)$.
\end{proof}

\section{Discussion and Conclusion} \label{conclusion}
\vspace{-.5em}
In this paper we have presented two deterministic distributed approximation algorithms for the ST problem in the CCM. The first one constructs a $2(1 - 1/\ell)$-approximate ST in $\tilde{O}(n^{1/3})$ rounds and $\tilde{O}(n^{7/3})$ messages and the polylogarthmic factors involved with each of the round and message complexities is $\log n$. The second one computes a $2(1 - 1/\ell)$-approximate ST in $O(S + \log \log n)$ rounds and $O(S (n - t)^2 + n^2)$ messages. Note that if a graph has the property  $S = \omega(n^{1/3} \log n)$, the first algorithm shows a better performance in terms of round complexity in contrast to the second one. On the other hand if a graph has the property $S = \tilde{o}(n^{1/3})$, then the second algorithm outperforms the first one in terms of the round complexity. Furthermore we have also observed that for constant or sufficiently small shortest path diameter networks (where $S = O(\log \log n)$) the second algorithm computes a $2(1 - 1/\ell)$-approximate ST in $O(\log \log n)$ rounds and $\tilde{O}(n^2)$ messages in the CCM. This result almost coincide with the best known deterministic result for MST construction in the CCM due to Lotker et al. \cite{Lotker:2005:MST:1085579.1085591} and the approximation factor of the resultant ST is at most $2(1 - 1/\ell)$ of the optimal.

It is known that in the centralized setting, the ST problem can be approximated upto a factor of $\ln (4) + \epsilon \approx 1.386 + \epsilon$ (for $\epsilon > 0$) of the optimal \cite{Byrka:2010:ILA:1806689.1806769}. Pemmaraju and Sardeshmukh \cite{pemmaraju_et_al:LIPIcs:2016:6882} showed that there exists a randomized algorithm that computes an MST in $O(\log^* n)$ rounds and $o(m)$ messages in the CCM. Recently  Jurdzi\'{n}ski and Nowicki \cite{Jurdzinski:2018:MOR:3174304.3175472} achieved a randomized algorithm  in the CCM that computes an MST in $O(1)$ rounds. Since the ST problem is a generalized version of the MST problem, there is an open research direction on  improvement of the approximation factor, the round and the message complexities in the CCM.




\newpage

\begin{thebibliography}{100}
\expandafter\ifx\csname natexlab\endcsname\relax\def\natexlab#1{#1}\fi
\providecommand{\url}[1]{\texttt{#1}}
\providecommand{\href}[2]{#2}
\providecommand{\path}[1]{#1}
\providecommand{\DOIprefix}{doi:}
\providecommand{\ArXivprefix}{arXiv:}
\providecommand{\URLprefix}{URL: }
\providecommand{\Pubmedprefix}{pmid:}
\providecommand{\doi}[1]{\href{http://dx.doi.org/#1}{\path{#1}}}
\providecommand{\Pubmed}[1]{\href{pmid:#1}{\path{#1}}}
\providecommand{\bibinfo}[2]{#2}
\ifx\xfnm\relax \def\xfnm[#1]{\unskip,\space#1}\fi
\bibitem[{Lotker et~al.(2005)Lotker, Patt-Shamir, Pavlov, and
  Peleg}]{Lotker:2005:MST:1085579.1085591}
\bibinfo{author}{Z.~Lotker}, \bibinfo{author}{B.~Patt-Shamir},
  \bibinfo{author}{E.~Pavlov}, \bibinfo{author}{D.~Peleg},
\newblock \bibinfo{title}{{Minimum-Weight Spanning Tree Construction in O(Log
  Log N) Communication Rounds}},
\newblock \bibinfo{journal}{SIAM Journal on Computing} \bibinfo{volume}{35}
  (\bibinfo{year}{2005}) \bibinfo{pages}{120--131}.
  \DOIprefix\doi{10.1137/S0097539704441848}.
\bibitem[{Linial(1992)}]{Linial:1992:LDG:130563.130578}
\bibinfo{author}{N.~Linial},
\newblock \bibinfo{title}{Locality in distributed graph algorithms},
\newblock \bibinfo{journal}{SIAM Journal on Computing} \bibinfo{volume}{21}
  (\bibinfo{year}{1992}) \bibinfo{pages}{193--201}.
  \DOIprefix\doi{10.1137/0221015}.
\bibitem[{Peleg and Rubinovich(2000)}]{Peleg2000ANL}
\bibinfo{author}{D.~Peleg}, \bibinfo{author}{V.~Rubinovich},
\newblock \bibinfo{title}{A near-tight lower bound on the time complexity of
  distributed minimum-weight spanning tree construction},
\newblock \bibinfo{journal}{SIAM Journal on Computing} \bibinfo{volume}{30}
  (\bibinfo{year}{2000}) \bibinfo{pages}{1427--1442}.
  \DOIprefix\doi{10.1137/S0097539700369740}.
\bibitem[{Karp(1972)}]{DBLP:conf/coco/Karp72}
\bibinfo{author}{R.~M. Karp},
\newblock \bibinfo{title}{Reducibility {A}mong {C}ombinatorial {P}roblems},
\newblock in: \bibinfo{booktitle}{Proceedings of a Symposium on the Complexity
  of Computer Computations}, \bibinfo{year}{1972}, pp.
  \bibinfo{pages}{85--103}. \DOIprefix\doi{10.1007/978-1-4684-2001-2_9}.
\bibitem[{Byrka et~al.(2010)Byrka, Grandoni, Rothvo\ss, and
  Sanit\`{a}}]{Byrka:2010:ILA:1806689.1806769}
\bibinfo{author}{J.~Byrka}, \bibinfo{author}{F.~Grandoni},
  \bibinfo{author}{T.~Rothvo\ss}, \bibinfo{author}{L.~Sanit\`{a}},
\newblock \bibinfo{title}{{An Improved LP-based Approximation for Steiner
  Tree}},
\newblock in: \bibinfo{booktitle}{Proceedings of the Forty-second ACM Symposium
  on Theory of Computing}, STOC '10, \bibinfo{year}{2010}, pp.
  \bibinfo{pages}{583--592}. \DOIprefix\doi{10.1145/1806689.1806769}.
\bibitem[{Chleb\'{\i}k and
  Chleb\'{\i}kov\'{a}(2008)}]{chlebik:2008:STP:1414105.1414423}
\bibinfo{author}{M.~Chleb\'{\i}k}, \bibinfo{author}{J.~Chleb\'{\i}kov\'{a}},
\newblock \bibinfo{title}{The {S}teiner tree problem on graphs:
  Inapproximability results},
\newblock \bibinfo{journal}{Theoretical Computer Science} \bibinfo{volume}{406}
  (\bibinfo{year}{2008}) \bibinfo{pages}{207 -- 214}.
  \DOIprefix\doi{10.1016/j.tcs.2008.06.046}.
\bibitem[{Hauptmann and Karpinski(2015)}]{Hauptman-Karpinaski}
\bibinfo{author}{M.~Hauptmann}, \bibinfo{author}{M.~Karpinski},
\newblock \bibinfo{title}{{A Compendium on {S}teiner Tree Problems}},
\newblock
  \bibinfo{journal}{http://theory.cs.uni-bonn.de/info5/steinerkompendium/netcompendium.html}
   (\bibinfo{year}{2015}).
\bibitem[{Hegeman et~al.(2014)Hegeman, Pemmaraju, and
  Sardeshmukh}]{10.1007/978-3-662-45174-8_35}
\bibinfo{author}{J.~W. Hegeman}, \bibinfo{author}{S.~V. Pemmaraju},
  \bibinfo{author}{V.~B. Sardeshmukh},
\newblock \bibinfo{title}{Near-constant-time distributed algorithms on a
  congested clique},
\newblock in: \bibinfo{booktitle}{Proceedings of the 28th International
  Conference on Distributed Computing}, DISC'14, \bibinfo{year}{2014}, pp.
  \bibinfo{pages}{514--530}. \DOIprefix\doi{10.1007/978-3-662-45174-8_35}.
\bibitem[{Ghaffari and Parter(2016)}]{Ghaffari:2016:MLR:2933057.2933103}
\bibinfo{author}{M.~Ghaffari}, \bibinfo{author}{M.~Parter},
\newblock \bibinfo{title}{{MST in Log-Star Rounds of Congested Clique}},
\newblock in: \bibinfo{booktitle}{Proceedings of the 2016 ACM Symposium on
  Principles of Distributed Computing}, PODC '16, \bibinfo{year}{2016}, pp.
  \bibinfo{pages}{19--28}. \DOIprefix\doi{10.1145/2933057.2933103}.
\bibitem[{Jurdzi\'{n}ski and
  Nowicki(2018)}]{Jurdzinski:2018:MOR:3174304.3175472}
\bibinfo{author}{T.~Jurdzi\'{n}ski}, \bibinfo{author}{K.~Nowicki},
\newblock \bibinfo{title}{{MST in O(1) Rounds of Congested Clique}},
\newblock in: \bibinfo{booktitle}{Proceedings of the Twenty-Ninth Annual
  ACM-SIAM Symposium on Discrete Algorithms}, SODA '18, \bibinfo{year}{2018},
  pp. \bibinfo{pages}{2620--2632}. \DOIprefix\doi{10.1137/1.9781611975031.167}.
\bibitem[{Berns et~al.(2012)Berns, Hegeman, and
  Pemmaraju}]{Berns:2012:SDA:2359888.2359973}
\bibinfo{author}{A.~Berns}, \bibinfo{author}{J.~Hegeman},
  \bibinfo{author}{S.~V. Pemmaraju},
\newblock \bibinfo{title}{{Super-fast Distributed Algorithms for Metric
  Facility Location}},
\newblock in: \bibinfo{booktitle}{Proceedings of the 39th International
  Colloquium Conference on Automata, Languages, and Programming - Volume Part
  II}, ICALP'12, \bibinfo{year}{2012}, pp. \bibinfo{pages}{428--439}.
  \DOIprefix\doi{10.1007/978-3-642-31585-5_39}.
\bibitem[{Gehweiler et~al.(2014)Gehweiler, Lammersen, and
  Sohler}]{Gehweiler2014}
\bibinfo{author}{J.~Gehweiler}, \bibinfo{author}{C.~Lammersen},
  \bibinfo{author}{C.~Sohler},
\newblock \bibinfo{title}{{A Distributed O(1)-Approximation Algorithm for the
  Uniform Facility Location Problem}},
\newblock \bibinfo{journal}{Algorithmica} \bibinfo{volume}{68}
  (\bibinfo{year}{2014}) \bibinfo{pages}{643--670}.
  \DOIprefix\doi{10.1007/s00453-012-9690-y}.
\bibitem[{Censor-Hillel et~al.(2015)Censor-Hillel, Kaski, Korhonen, Lenzen,
  Paz, and Suomela}]{Censor-Hillel:2015:AMC:2767386.2767414}
\bibinfo{author}{K.~Censor-Hillel}, \bibinfo{author}{P.~Kaski},
  \bibinfo{author}{J.~H. Korhonen}, \bibinfo{author}{C.~Lenzen},
  \bibinfo{author}{A.~Paz}, \bibinfo{author}{J.~Suomela},
\newblock \bibinfo{title}{{Algebraic Methods in the Congested Clique}},
\newblock in: \bibinfo{booktitle}{Proceedings of the 2015 ACM Symposium on
  Principles of Distributed Computing}, PODC '15, \bibinfo{year}{2015}, pp.
  \bibinfo{pages}{143--152}. \DOIprefix\doi{10.1145/2767386.2767414}.
\bibitem[{Holzer and Pinsker(2016)}]{holzer_et_al:LIPIcs:2016:6597}
\bibinfo{author}{S.~Holzer}, \bibinfo{author}{N.~Pinsker},
\newblock \bibinfo{title}{{Approximation of Distances and Shortest Paths in the
  Broadcast Congest Clique}},
\newblock in: \bibinfo{booktitle}{19th International Conference on Principles
  of Distributed Systems (OPODIS 2015)}, volume~\bibinfo{volume}{46},
  \bibinfo{year}{2016}, pp. \bibinfo{pages}{1--16}.
  \DOIprefix\doi{10.4230/LIPIcs.OPODIS.2015.6}.
\bibitem[{Nanongkai(2014)}]{Nanongkai:2014:DAA:2591796}
\bibinfo{author}{D.~Nanongkai},
\newblock \bibinfo{title}{Distributed approximation algorithms for weighted
  shortest paths},
\newblock in: \bibinfo{booktitle}{Proceedings of the Forty-sixth Annual ACM
  Symposium on Theory of Computing}, STOC '14, \bibinfo{year}{2014}, pp.
  \bibinfo{pages}{565--573}. \DOIprefix\doi{10.1145/2591796.2591850}.
\bibitem[{Drucker et~al.(2012)Drucker, Kuhn, and
  Oshman}]{Drucker:2012:CCD:2332432.2332443}
\bibinfo{author}{A.~Drucker}, \bibinfo{author}{F.~Kuhn},
  \bibinfo{author}{R.~Oshman},
\newblock \bibinfo{title}{The communication complexity of distributed task
  allocation},
\newblock in: \bibinfo{booktitle}{Proceedings of the 2012 ACM Symposium on
  Principles of Distributed Computing}, PODC '12, \bibinfo{year}{2012}, pp.
  \bibinfo{pages}{67--76}. \DOIprefix\doi{10.1145/2332432.2332443}.
\bibitem[{Dolev et~al.(2012)Dolev, Lenzen, and
  Peled}]{Dolev:2012:TTA:2427873.2427892}
\bibinfo{author}{D.~Dolev}, \bibinfo{author}{C.~Lenzen},
  \bibinfo{author}{S.~Peled},
\newblock \bibinfo{title}{{``Tri, Tri Again'': Finding Triangles and Small
  Subgraphs in a Distributed Setting}},
\newblock in: \bibinfo{booktitle}{Proceedings of the 26th International
  Conference on Distributed Computing}, DISC'12, \bibinfo{year}{2012}, pp.
  \bibinfo{pages}{195--209}. \DOIprefix\doi{10.1007/978-3-642-33651-5_14}.
\bibitem[{Patt-Shamir and
  Teplitsky(2011)}]{Patt-Shamir:2011:RCD:1993806.1993851}
\bibinfo{author}{B.~Patt-Shamir}, \bibinfo{author}{M.~Teplitsky},
\newblock \bibinfo{title}{The round complexity of distributed sorting: Extended
  abstract},
\newblock in: \bibinfo{booktitle}{Proceedings of the 30th Annual ACM
  SIGACT-SIGOPS Symposium on Principles of Distributed Computing}, PODC '11,
  \bibinfo{year}{2011}, pp. \bibinfo{pages}{249--256}.
  \DOIprefix\doi{10.1145/1993806.1993851}.
\bibitem[{Lenzen(2013)}]{Lenzen:2013:ODR:2484239.2501983}
\bibinfo{author}{C.~Lenzen},
\newblock \bibinfo{title}{Optimal deterministic routing and sorting on the
  congested clique},
\newblock in: \bibinfo{booktitle}{Proceedings of the 2013 ACM Symposium on
  Principles of Distributed Computing}, PODC '13, \bibinfo{year}{2013}, pp.
  \bibinfo{pages}{42--50}. \DOIprefix\doi{10.1145/2484239.2501983}.
\bibitem[{Pai and Pemmaraju(2019)}]{DBLP:journals/corr/abs-1905-09016}
\bibinfo{author}{S.~Pai}, \bibinfo{author}{S.~V. Pemmaraju},
\newblock \bibinfo{title}{Connectivity lower bounds in broadcast congested
  clique},
\newblock \bibinfo{journal}{CoRR}  (\bibinfo{year}{2019}). \URLprefix
  \url{http://arxiv.org/abs/1905.09016}.
\bibitem[{Drucker et~al.(2014)Drucker, Kuhn, and
  Oshman}]{Drucker:2014:PCC:2611462.2611493}
\bibinfo{author}{A.~Drucker}, \bibinfo{author}{F.~Kuhn},
  \bibinfo{author}{R.~Oshman},
\newblock \bibinfo{title}{On the power of the congested clique model},
\newblock in: \bibinfo{booktitle}{Proceedings of the 2014 ACM Symposium on
  Principles of Distributed Computing}, PODC '14, \bibinfo{year}{2014}, pp.
  \bibinfo{pages}{367--376}. \DOIprefix\doi{10.1145/2611462.2611493}.
\bibitem[{Chen et~al.(1993)Chen, Houle, and Kuo}]{GENHUEY199373}
\bibinfo{author}{G.-H. Chen}, \bibinfo{author}{M.~E. Houle},
  \bibinfo{author}{M.-T. Kuo},
\newblock \bibinfo{title}{The {S}teiner problem in distributed computing
  systems},
\newblock \bibinfo{journal}{Information Sciences} \bibinfo{volume}{74}
  (\bibinfo{year}{1993}) \bibinfo{pages}{73--96}.
  \DOIprefix\doi{10.1016/0020-0255(93)90128-9}.
\bibitem[{Chalermsook and Fakcharoenphol(2005)}]{PC_JF_2005}
\bibinfo{author}{P.~Chalermsook}, \bibinfo{author}{J.~Fakcharoenphol},
\newblock \bibinfo{title}{Simple distributed algorithms for approximating
  minimum {S}teiner trees},
\newblock in: \bibinfo{booktitle}{Proceedings of the 11th Annual International
  Conference on Computing and Combinatorics}, COCOON '05, \bibinfo{year}{2005},
  pp. \bibinfo{pages}{380--389}. \DOIprefix\doi{10.1007/11533719_39}.
\bibitem[{Khan et~al.(2008)Khan, Kuhn, Malkhi, Pandurangan, and
  Talwar}]{MK_FK_DM_GP_KT_2008}
\bibinfo{author}{M.~Khan}, \bibinfo{author}{F.~Kuhn},
  \bibinfo{author}{D.~Malkhi}, \bibinfo{author}{G.~Pandurangan},
  \bibinfo{author}{K.~Talwar},
\newblock \bibinfo{title}{Efficient distributed approximation algorithms via
  probabilistic tree embeddings},
\newblock in: \bibinfo{booktitle}{Proceedings of the Twenty-seventh ACM
  Symposium on Principles of Distributed Computing}, PODC '08,
  \bibinfo{year}{2008}, pp. \bibinfo{pages}{263--272}.
  \DOIprefix\doi{10.1145/1400751.1400787}.
\bibitem[{Lenzen and Patt-Shamir(2015)}]{Lenzen:2015:FPD}
\bibinfo{author}{C.~Lenzen}, \bibinfo{author}{B.~Patt-Shamir},
\newblock \bibinfo{title}{Fast partial distance estimation and applications},
\newblock in: \bibinfo{booktitle}{Proceedings of the 2015 ACM Symposium on
  Principles of Distributed Computing}, PODC '15, \bibinfo{year}{2015}, pp.
  \bibinfo{pages}{153--162}. \DOIprefix\doi{10.1145/2767386.2767398}.
\bibitem[{Saikia and Karmakar(2019)}]{Saikia-Karmakar2019}
\bibinfo{author}{P.~Saikia}, \bibinfo{author}{S.~Karmakar},
\newblock \bibinfo{title}{A simple $2(1-1/\ell)$ factor distributed
  approximation algorithm for {S}teiner tree in the congest model},
\newblock in: \bibinfo{booktitle}{Proceedings of the 20th International
  Conference on Distributed Computing and Networking}, ICDCN '19,
  \bibinfo{year}{2019}, pp. \bibinfo{pages}{41--50}.
  \DOIprefix\doi{10.1145/3288599.3288629}.
\bibitem[{Khan and Pandurangan(2008)}]{Khan2008}
\bibinfo{author}{M.~Khan}, \bibinfo{author}{G.~Pandurangan},
\newblock \bibinfo{title}{A fast distributed approximation algorithm for
  minimum spanning trees},
\newblock \bibinfo{journal}{Distributed Computing} \bibinfo{volume}{20}
  (\bibinfo{year}{2008}) \bibinfo{pages}{391--402}.
  \DOIprefix\doi{10.1007/s00446-007-0047-8}.
\bibitem[{Hegeman et~al.(2015)Hegeman, Pandurangan, Pemmaraju, Sardeshmukh, and
  Scquizzato}]{Hegeman:2015:TOB:2767386.2767434}
\bibinfo{author}{J.~W. Hegeman}, \bibinfo{author}{G.~Pandurangan},
  \bibinfo{author}{S.~V. Pemmaraju}, \bibinfo{author}{V.~B. Sardeshmukh},
  \bibinfo{author}{M.~Scquizzato},
\newblock \bibinfo{title}{{Toward Optimal Bounds in the Congested Clique: Graph
  Connectivity and MST}},
\newblock in: \bibinfo{booktitle}{Proceedings of the 2015 ACM Symposium on
  Principles of Distributed Computing}, PODC '15, \bibinfo{year}{2015}, pp.
  \bibinfo{pages}{91--100}. \DOIprefix\doi{10.1145/2767386.2767434}.
\bibitem[{Pemmaraju and Sardeshmukh(2016)}]{pemmaraju_et_al:LIPIcs:2016:6882}
\bibinfo{author}{S.~V. Pemmaraju}, \bibinfo{author}{V.~B. Sardeshmukh},
\newblock \bibinfo{title}{{Super-Fast MST Algorithms in the Congested Clique
  Using o(m) Messages}},
\newblock in: \bibinfo{booktitle}{36th IARCS Annual Conference on Foundations
  of Software Technology and Theoretical Computer Science (FSTTCS 2016)},
  \bibinfo{year}{2016}, pp. \bibinfo{pages}{47:1--47:15}.
  \DOIprefix\doi{10.4230/LIPIcs.FSTTCS.2016.47}.
\bibitem[{Kou et~al.(1981)Kou, Markowsky, and
  Berman}]{Kou:1981:FAS:2697742.2698014}
\bibinfo{author}{L.~Kou}, \bibinfo{author}{G.~Markowsky},
  \bibinfo{author}{L.~Berman},
\newblock \bibinfo{title}{A fast algorithm for {S}teiner trees},
\newblock \bibinfo{journal}{Acta Informatica} \bibinfo{volume}{15}
  (\bibinfo{year}{1981}) \bibinfo{pages}{141--145}.
  \DOIprefix\doi{10.1007/BF00288961}.
\bibitem[{Wu et~al.(1986)Wu, Widmayer, and Wong}]{Wu1986}
\bibinfo{author}{Y.~F. Wu}, \bibinfo{author}{P.~Widmayer},
  \bibinfo{author}{C.~K. Wong},
\newblock \bibinfo{title}{A faster approximation algorithm for the {S}teiner
  problem in graphs},
\newblock \bibinfo{journal}{Acta Informatica} \bibinfo{volume}{23}
  (\bibinfo{year}{1986}) \bibinfo{pages}{223--229}.
  \DOIprefix\doi{10.1007/BF00289500}.
\bibitem[{Bachrach et~al.(2019)Bachrach, Censor{-}Hillel, Dory, Efron,
  Leitersdorf, and Paz}]{DBLP:journals/corr/abs-1905-10284}
\bibinfo{author}{N.~Bachrach}, \bibinfo{author}{K.~Censor{-}Hillel},
  \bibinfo{author}{M.~Dory}, \bibinfo{author}{Y.~Efron},
  \bibinfo{author}{D.~Leitersdorf}, \bibinfo{author}{A.~Paz},
\newblock \bibinfo{title}{Hardness of distributed optimization},
\newblock \bibinfo{journal}{CoRR} \bibinfo{volume}{abs/1905.10284}
  (\bibinfo{year}{2019}). \URLprefix \url{http://arxiv.org/abs/1905.10284}.
\bibitem[{Das~Sarma et~al.(2011)Das~Sarma, Holzer, Kor, Korman, Nanongkai,
  Pandurangan, Peleg, and Wattenhofer}]{DasSarma:2011:DVH:1993636.1993686}
\bibinfo{author}{A.~Das~Sarma}, \bibinfo{author}{S.~Holzer},
  \bibinfo{author}{L.~Kor}, \bibinfo{author}{A.~Korman},
  \bibinfo{author}{D.~Nanongkai}, \bibinfo{author}{G.~Pandurangan},
  \bibinfo{author}{D.~Peleg}, \bibinfo{author}{R.~Wattenhofer},
\newblock \bibinfo{title}{Distributed verification and hardness of distributed
  approximation},
\newblock in: \bibinfo{booktitle}{Proceedings of the Forty-third Annual ACM
  Symposium on Theory of Computing}, STOC '11, \bibinfo{year}{2011}, pp.
  \bibinfo{pages}{363--372}. \DOIprefix\doi{10.1145/1993636.1993686}.
\bibitem[{Kutten et~al.(2015)Kutten, Pandurangan, Peleg, Robinson, and
  Trehan}]{Kutten:2015:CUL:2742144.2699440}
\bibinfo{author}{S.~Kutten}, \bibinfo{author}{G.~Pandurangan},
  \bibinfo{author}{D.~Peleg}, \bibinfo{author}{P.~Robinson},
  \bibinfo{author}{A.~Trehan},
\newblock \bibinfo{title}{On the complexity of universal leader election},
\newblock \bibinfo{journal}{Journal of the ACM} \bibinfo{volume}{62}
  (\bibinfo{year}{2015}) \bibinfo{pages}{7:1--7:27}.
  \DOIprefix\doi{10.1145/2699440}.
\bibitem[{Peleg(2000)}]{Peleg_2000}
\bibinfo{author}{D.~Peleg},
\newblock \bibinfo{title}{{Distributed Computing: A Locality-Sensitive
  Approach}},
\newblock \bibinfo{journal}{SIAM : Discrete Mathematics and Applications}
  (\bibinfo{year}{2000}). \DOIprefix\doi{10.1137/1.9780898719772}.
\bibitem[{Fischer and Meyer(1971)}]{Fischer:1971:BMM:1446293.1446319}
\bibinfo{author}{M.~J. Fischer}, \bibinfo{author}{A.~R. Meyer},
\newblock \bibinfo{title}{Boolean matrix multiplication and transitive
  closure},
\newblock in: \bibinfo{booktitle}{Proceedings of the 12th Annual Symposium on
  Switching and Automata Theory}, SWAT '71, \bibinfo{year}{1971}, pp.
  \bibinfo{pages}{129--131}. \DOIprefix\doi{10.1109/SWAT.1971.4}.
\bibitem[{Munro(1971)}]{Munro:1971:EDT:2598952.2599354}
\bibinfo{author}{I.~Munro},
\newblock \bibinfo{title}{Efficient determination of the transitive closure of
  a directed graph},
\newblock \bibinfo{journal}{Information Processing Letters} \bibinfo{volume}{1}
  (\bibinfo{year}{1971}) \bibinfo{pages}{56--58}.
  \DOIprefix\doi{10.1016/0020-0190(71)90006-8}.

\end{thebibliography}

\appendix

\section{Censor-Hillel et al.'s APSP algorithm in the CCM \cite{Censor-Hillel:2015:AMC:2767386.2767414}} \label{Censor-Hillel-all-algo}
\vspace{-.5em}
Censor-Hillel et al. showed that the APSP of a given graph $G = (V, E, w)$ can be computed via iterated squaring of the weight matrix over the min-plus semiring \cite{Fischer:1971:BMM:1446293.1446319,Munro:1971:EDT:2598952.2599354}.  Let $W$ be the weight matrix of size $n \times n$ for a given graph $G$, where $n = |V|$. Then the distance product which is also known as the min-plus product is defined as follows.

$(W \star W)_{uv} = W^2_{uv} = \smash{\displaystyle\min_{w}} (W_{uw} + W_{wv})$

where $u, v, w \in V$. The $n^{th}$ distance product denoted by $W^{n}$ produces the actual shortest path distances in $G$, i.e., $\delta(u, v) = W^{n}_{uv}$ for each pair of vertices $u, v \in V$. The distance product $W^n$ can be computed by iteratively squaring $W$ for $\lceil \log n \rceil$ times. The distributed squaring of a weight matrix over the min-plus semiring is similar to the parallel 3D matrix multiplication. The squaring of $W$ over the min-plus semiring involves $n^3$ element wise additions of the form $W_{uw} + W_{wv}$, where $u, v, w \in V$. This is viewed as a cube of size $n \times n \times n$ where each point in the cube represents an addition operation. If we partition the cube of size $n \times n \times n$ into $n$ sub-cubes of equal sizes then each of them gets $n^{2/3} \times n^{2/3} \times n^{2/3}$ points of the cube. This indicates that each node $v \in V$ needs to perform $n^{2/3} \times n^{2/3} \times n^{2/3}$ addition operations. Note that the weight matrix $W$ has  $n^2$ entries corresponding to the edge weights of the input graph.\footnote{If an edge $(u, v)$ is not in the input graph $G$, then $w((u, v))$ is considered equal to $\infty$.} This $n^2$ entries are distributed uniformly among $n$ nodes in such a way that each node $v \in V$ gets two sub-matrices of sizes $n^{2/3} \times n^{2/3}$. This ensures that each node needs to compute the min-plus product of the matrix size $n^{2/3} \times n^{2/3}$ only.

Now we briefly describe the APSP algorithm due to Censor-Hillel et al. \cite{Censor-Hillel:2015:AMC:2767386.2767414}. Initially each node $u \in V$ knows the edge weights $w((u, v))$ for each $v \in V$. Here it is assumed that $n^{1/3}$ is a positive integer, where $n = |V|$. Each node $v \in V$ partitions its input  into $n^{1/3}$ blocks, each one has $n^{2/3}$ entries. Now each node $v \in V$ sends its each block to $2n^{1/3}$ nodes in such a way that each node $v \in V$ receives two sub-matrices, each of them has size $n^{2/3} \times n^{2/3}$. Since each block has size $n^{2/3}$ and there are $2n^{2/3}$ recipients, a total of $2n^{4/3}$ messages are sent by each node. This distribution is performed in $O(n^{1/3})$ rounds and $O(n^{7/3})$ messages by using the deterministic routing scheme of Lenzen \cite{Lenzen:2013:ODR:2484239.2501983}. Then each node computes the min-plus product from the two known sub-matrices. After that all the resulting min-plus products are re-distributed among $n$ nodes in such a way that each node $v \in V$ receives $n^{4/3}$ values. This distribution of the min-plus products is also performed in $O(n^{1/3})$ rounds $O(n^{7/3})$ messages by using the deterministic routing scheme of Lenzen \cite{Lenzen:2013:ODR:2484239.2501983}. The received $n^{4/3}$ messages by a node $v \in V$ contain the information related to the row $v$ of the resulting min-plus product $W^2 = W \star W$ from which $v$ locally computes the values of its own row with respect to $W^2$. With this the first iteration ends. Considering that each of the node IDs and edge weights can be encoded by using $O(\log)$ bits, the round and message complexities of the this iteration are $O(n^{1/3})$ and $O(n^{7/3})$ respectively. Repetition of the above procedure for $\lceil \log n \rceil$ times on each of the resulting min-plus product guarantees the final output $W^n$ which is the required APSP for a given graph $G$. The overall round and message complexities of this algorithm are $O(n^{1/3} \log n)$ and $O(n^{7/3} \log n)$ respectively.

\medskip
{\bf Routing table construction}. The task of routing table construction concerns computing local tables at all the nodes of a network in which each node $u$, when given a destination node $v$, knows the next hop through which $u$ is connected to $v$ by the shortest path distance in the network. Specifically the routing table entry $R[u, v] = w \in V$ is a node such that $(u, w) \in E$ and $w$ lies on a shortest path from $u$ to $v$. Censor-Hillel et al. \cite{Censor-Hillel:2015:AMC:2767386.2767414} showed that iterated squaring of the weight matrix over the min-plus semiring can also be used to construct the routing table for each node $v \in V$. The min-plus product $W \star W$ provides a \emph{witness matrix} $Q$ such that if $Q_{uv} = w$, then $(W \star W)_{uv} = W_{uw} + W_{wv}$. Whenever the iterated squaring algorithm computes the min-plus product $W^{2i} = W^i \star W^i$, then using the witness matrix $Q$ the routing table $R$ is updated to $R[u, v] = R[u, Q_{uv}]$ for each $u, v \in V$ with $W^{2i}_{uv} < W^i_{uv}$.

\section{Lotker et al's MST algorithm in the CCM \cite{Lotker:2005:MST:1085579.1085591}} \label{Lotker-algo}
The algorithm operates in phases. Each phase takes $O(1)$ rounds. In each phase $k \geq 0$ it maintains a set of clusters $\mathcal{F}^k = \{\mathit{F}^k_1, \mathit{F}^k_2,.....\mathit{F}^k_p\}$, $\bigcup_{i} \mathit{F_i^k = V}$, where $\mathit{V}$ is the vertex set of the input graph $\mathit{G}$.  For the sake of simplicity it is assumed that at the beginning of phase $1$, the end of the imaginary phase $0$, with the cluster set $\mathcal{F}^0 = \{\mathit{F}^0_1, \mathit{F}^0_2,.....\mathit{F}^0_n\}$ is known, where $\mathit{F}^0_i = \{v_i\}$ for every $1 \leq i \leq n$. For each cluster $\mathit{F} \in \mathcal{F}^k$, the algorithm selects a spanning subtree $\mathit{T(F)}$. At the beginning of phase $k > 0$, the cluster set $\mathcal{F}^{k - 1}$ and the corresponding subtree collection $\mathcal{T}^{k - 1} = \{ \mathit{T(F)} | \mathit{F} \in \mathcal{F}^{k - 1} \}$, including the weights of the edges in those subtrees, are known to every vertex in the graph. Whenever the algorithm terminates each node in $\mathit{V}$ knows  all the $n - 1$ edges in the MST of the given graph $\mathit{G}$.

The outline of a phase $k > 0$ is as follows. Initially each cluster $\mathit{C} \in \mathcal{F}^{k - 1}$ is contracted to a vertex $\mathit{v_C}$. All these contracted vertices together form a smaller logical graph denoted by $\mathit{\hat{G}}$. The operation of each vertex $\mathit{v_C}$ is carried out by a special node called the leader of $\mathit{C}$ denoted by $\mathit{l(C)}$.  Let $\mathit{N}$ be the minimum size cluster in $\mathcal{F}^{k - 1}$. At the beginning of phase $k > 0$, $\mathit{N}$ (or more) the members of each cluster $\mathit{C} \in \mathcal{F}^{k - 1}$ collect $\mathit{N}$ lightest edges connecting their cluster $\mathit{C}$ to other clusters in $\mathcal{F}^{k - 1} \setminus \{\mathit{C}\}$. Then each cluster in $\mathcal{F}^{k - 1}$ sends these $\mathit{N}$ (or more) lightest edges to a special node $v_0$ (node with the lowest ID) of the graph by appropriately sharing the workload among the nodes of the cluster. Now the node $v_0$ has a partial picture of the logical graph $\mathit{\hat{G}}$, consisting of all the contracted vertices $\mathit{v_C}$ but only some of the edges connecting them, specifically $\mathit{N}$ lightest edges emanating from each node of $\mathit{\hat{G}}$ to $\mathit{N}$ different nodes. On the basis of these received information $v_0$ performs (locally) fragment merging operations. For that the known edges are sorted in non decreasing order of their weights. Then from the non decreasing ordered edge set add edges to $\mathit{\hat{G}}$ \ that are the minimum weight outgoing edge (MWOE) of one of two fragments they connect as long as merging is perfectly safe. The \emph{safety rule} is as follows:  It is perfectly safe to continue merging a fragment $\mathit{F}$ in the logical graph $\mathit{\hat{G}}$ as long as the $\mathit{N}$ lightest edges of each vertex $\mathit{v_C}$ in $\mathit{F}$ are not inspected. It is shown that the safety rule allows to grow each of the fragments to contain at least $\mathit{N} + 1$ vertices of $\mathit{\hat{G}}$. This ensures that the size of each of the clusters in the next phase will be at least $\mathit{N}^2$. This is a \textit{quadratic} growth of the clusters. Finally $v_0$ sends out the locally known identity of all the edges that are newly chosen. Note that the number of such edges can be at most $n - 1$. The special node $v_0$ performs this in $O(1)$ rounds by sending each edge to a different intermediate node, which will broadcast that edge to all other nodes.

It is clear that at the end of a phase $k > 0$, each of the cluster sizes is at least $2^{2^{k-1}}$. Whenever the algorithm terminates, there exists only one cluster in the cluster set which is the required MST of $\mathit{G}$, i.e., $|\mathcal{F}^k| = 1$. Therefore ${2^{2^{k - 1} }} = n \Rightarrow k = \log \log n + 1$. This ensures that the algorithm terminates in $O(\log \log n)$ rounds, whereas the message complexity is $O(n^2)$.




\end{document}